\documentclass[12pt]{amsart}
\usepackage{graphicx}
\usepackage{setspace}
\usepackage{natbib}
\usepackage{hyperref}
\usepackage[foot]{amsaddr}
\usepackage{amsmath}	
\usepackage{bm}
\allowdisplaybreaks
\newtheorem{thm}{Theorem}

\newtheorem{lem}[thm]{Lemma}

\newtheorem{assum}{Assumption}
\theoremstyle{definition}

\newtheorem*{algorithm}{Algorithm}
\theoremstyle{remark}
\newtheorem{rem}{Remark}

\numberwithin{equation}{section}
\newcommand{\norm}[1]{\left\Vert#1\right\Vert}

\DeclareMathOperator*{\argmax}{argmax}
\DeclareMathOperator*{\argmin}{argmin}

\title[Doubly Robust Uniform Confidence Band]{Doubly Robust Uniform Confidence Band for the Conditional Average Treatment Effect Function}
\thanks{We would like to thank Yanchun Jin for capable research assistance,
Takahiro Hoshino, Yu-Chin Hsu, Art\=uras Juodis, Kengo Kato,  Edward Kennedy, Taisuke Otsu, Dylan Small, seminar participants at various institutes, the co-editor and three anonymous referees for helpful comments, and Robert Lieli for providing the data to us.
Lee's work
was supported by the European Research Council (ERC-2009-StG-240910-ROMETA and ERC-2014-CoG-646917-ROMIA).
Okui's work was supported by the Japan Society of the Promotion of Science (KAKENHI 25285067, 25780151, 15H03329, 16K03598). 
Whang's work was supported by the SNU College of Social Science Research Grant.}

\date{This version: October 2016; the first draft: November 2014.}

\author[Lee]{Sokbae Lee$^{1,2}$ }
\address{$^1$Department of Economics, Columbia University, 1022 International Affairs Building,
420 West 118th Street, New York, NY 10027, USA.}
\address{$^2$Centre for Microdata Methods and Practice, Institute for Fiscal
Studies, 7 Ridgmount Street, London, WC1E 7AE, UK.}
\email{sl3841@columbia.edu}
\author[Okui]{Ryo Okui$^{3,4}$}
\address{$^3$Institute of Economic Research, Kyoto University, Yoshida-Honmachi
Sakyo, Kyoto, 606-8501, Japan.}
\address{$^4$Department of Economics, University of Gothenburg,
P.O. Box 640, SE-405 30 Gothenburg, Sweden.}
\email{okui.ryo.3@gmail.com}
\author[Whang]{Yoon-Jae Whang$^5$}
\address{$^5$Department of Economics, Seoul National University, 1
Gwanak-ro, Gwanak-gu, Seoul, 151-742, Republic of Korea.}
\email{whang@snu.ac.kr}

\begin{document}

\doublespacing

\begin{abstract}
 In this paper, we propose a doubly robust method to estimate the heterogeneity of the average treatment effect with respect to observed covariates of interest. 
We consider a situation where a large number of covariates are needed for identifying the average treatment effect but the covariates of interest for analyzing heterogeneity are of much lower dimension. 
Our proposed estimator is doubly robust and avoids the curse of dimensionality.
We propose a uniform confidence band that is easy to compute, and we illustrate its usefulness via Monte Carlo experiments and an application to the effects of smoking on birth weights. 
\bigskip

\bigskip



\end{abstract}

\maketitle

\clearpage

\section{Introduction}

In this paper, we propose a doubly robust method to estimate the heterogeneity of the average treatment effect
with respect to observed covariates of interest. To describe our methodology, 
we consider the potential outcome framework.
Let $Y_{1}$ and $Y_{0}$ be potential individual outcomes in two states, with
treatment and without treatment, respectively. For each individual, the observed outcome $%
Y $ is $Y=DY_{1}+(1-D)Y_{0}$, where $D$ denotes an indicator
variable for the treatment, with $D=0$ if an individual is not treated and $%
D=1$ if an individual is treated. We assume that independent and identically
distributed observations $\{(Y_{i},D_{i},\mathbf{Z}_{i}):i=1,\ldots ,n\}$ of $(Y,D,\mathbf{Z})$
are available, where $\mathbf{Z} \in \mathbb{R}^{p}$ denotes a $p$-dimensional vector of covariates. 

Suppose that a researcher is interested in evaluating the average treatment effect conditional on only a
subset of covariates $\mathbf{X}$,
 which is of a substantially lower dimension than $\mathbf{Z}$, 
 where $\mathbf{Z}\equiv (\mathbf{X}^{\top },\mathbf{V}^{\top })^{\top }\in \mathbb{R}^{d}\times
\mathbb{R}^{m}$, $p \equiv d+m$.
That is, we are interested in a case where $d \ll p$.

The main object of interest in this paper is the conditional average treatment effect function (CATEF); namely: 
\begin{align}\label{def:cater}
g(\mathbf{x}) \equiv\mathbb{E}[Y_{1}-Y_{0}|\mathbf{X}=\mathbf{x}].
\end{align}
When $d \geq 3$, it is difficult to plot $g(\mathbf{x})$, not to mention low precision due to the curse of dimensionality.
Hence, for practical reasons, we focus on the case that $d=1$ or $d=2$, while $p$ is often of a much higher dimension.

To achieve identification of the CATEF, we
assume that $Y_1$ and $Y_0$ are independent of 
$D$ conditional on $\mathbf{Z}$ (known as the unconfoundedness assumption):
\begin{align}\label{unconfound}
 (Y_1, Y_0) \perp D |\mathbf{Z},
\end{align}  
where $\perp$ denotes the independence.
For \eqref{unconfound} to be plausible in applications, applied researchers
tend to consider a large number of covariates $\mathbf{Z}$. 
Note that in our setup, the treatment may be confounded in the sense that the treatment assignment may not be independent of the potential outcome variables given $\mathbf{X}$ only.
To satisfy the unconfoundedness condition, a much larger set of conditioning variables $\mathbf{Z}$ needs to be employed.

Different roles of covariates between $\mathbf{X}$ and $\mathbf{V}$
are noted in the recent literature. 
For example, \cite{Ogburn-et-al-JRSSB} consider a similar issue in the context of 
the local average treatment effect (LATE) of \cite{Imbens:Angrist:94}.
\cite{Ogburn-et-al-JRSSB} emphasize that 
conditioning on a large number of covariates $\mathbf{Z}$ may be required to make it plausible that 
the binary instrument is valid. 
In their empirical example, \cite{Ogburn-et-al-JRSSB} revisit the analyses of 
\cite{Poterba1995} and \cite{Abadie2003} to examine whether participation in 401(k) pension plans increases household savings. 
In their example, the vector of covariates $\mathbf{Z}$ for the identifying assumption consists of 
income, age, marital status, and family size, whereas the variable of interest $\mathbf{X}$ is income.
 \cite{Abrevaya-et-al-JBES} also consider the case of investigating the effect of smoking during pregnancy on birth weights. 
They are interested in estimating \eqref{def:cater} with $\mathbf{X}$ being the age of mother; however, as noted in \cite{Abrevaya-et-al-JBES}, it is unlikely that conditioning only on the age of mother would achieve the unconfoundedness assumption with nonexperimental data. As a result,
it is necessary to consider a high-dimensional $\mathbf{Z}$, including the age of the mother. 

The fact that a high-dimensional $\mathbf{Z}$ needs to be employed for \eqref{unconfound} to be plausible in an application makes a fully
nonparametric estimation approach impractical because of the curse of
dimensionality. For example, the propensity score is not nonparametrically estimable 
in moderately sized samples, if the dimension of $\mathbf{Z}$ is high.
One obvious alternative is to use a parametric model for the propensity score; however, it may lead to misleading results if the parametric model is misspecified. 

With the aim of providing a practical method and, at the same time, reducing sensitivity to model misspecification, 
we propose to use a doubly robust method based on parametric regression and propensity score models. 
 Our estimator of the CATEF is doubly robust in the sense that it is consistent when at least one of the regression model and the propensity score model is correctly specified.
Specifically, we first estimate CATEF$(\mathbf{Z})$ using a doubly robust procedure:
we estimate a parametric regression model of the outcome on $\mathbf{Z}$ for each treatment status and a parametric model for the probability of selecting into the treatment given $\mathbf{Z}$; we then 
combine the parametric estimation results in a doubly robust fashion to construct an estimate of CATEF$(\mathbf{Z})$.
We then obtain an estimate of CATEF$(\mathbf{X})$ by adopting the local linear smoothing of CATEF$(\mathbf{Z})$.
As a result, we avoid high-dimensional smoothing with respect to 
$\mathbf{Z}$ but mitigate the problem of misspecification by both the doubly robust estimation and low-dimensional nonparametric smoothing with respect to $\mathbf{X}$.

We emphasize that we are  willing to assume parametric specifications for the propensity score and 
regression models as functions of $\mathbf{Z}$ to avoid the curse of dimensionality, but not for CATEF$(\mathbf{X})$.
One may consider parametric estimation of CATEF$(\mathbf{X})$, as \cite{Ogburn-et-al-JRSSB} estimate their LATE parameter
using least squares approximations. However, note that 
even if the parametric specification of CATEF$(\mathbf{Z})$ is correct, the resulting specification of CATEF$(\mathbf{X})$
may not be correctly specified  since, for example, $\mathbb{E}[ \mathbf{Z}| \mathbf{X}]$ is possibly highly nonlinear.
To avoid this misspecification, we estimate CATEF$(\mathbf{X})$ nonparametrically.

Because the CATEF is a functional parameter, as a tool of inference, we propose to use a uniform confidence band for the CATEF.
Our construction of the uniform confidence band is based on some analytic approximation
of the supremum of a Gaussian process using arguments built on \cite{Piterbarg:96}, combined with a Gaussian approximation result of 
\cite{CCK:2014} and an empirical process result of \cite{Ghosal/Sen/vanderVaart:00}.
Our method is simple to implement 
and does not rely on resampling techniques.

This paper contributes to the literature on doubly robust estimation by demonstrating that the doubly robust procedures are useful for estimating the CATEF.
In this paper, we focus on the so-called augmented inverse probability weighting estimator that was originally proposed by \cite{RobinsRotnitzkyZhao94} for the estimation of the mean \citep[see also][]{RobinsRotnitzky95, ScharfsteinRotnitzkyRobins99}.
Their estimator appears to be the first estimator to be recognized as being doubly robust.
Since then, many other alternative doubly robust estimators have been proposed in the literature.
For example, the inverse probability weighting regression adjustment estimator \citep{KangSchafer07, Wooldridge07, Wooldridge10} 
is widely known and has been implemented in statistical software packages.
See the introduction of \cite{Tan10} for a comprehensive summary of other doubly robust estimators.
Doubly robust estimators have been advocated for use in many different areas of application: See, for example, \cite{LuncefordDavidian04} for medicine, \cite{GlynnQuinn10} for political science, \cite{Wooldridge10} for economics, and \cite{SchaferKang08} for psychology.
There are also doubly robust estimators available for different settings including instrumental variables estimation \citep{Tan06, OkuiSmallTanRobins12}
and estimation under multivalued treatments \citep{DeryaUysal:15}.
It would not be difficult to extend our method to allow these other doubly robust estimators and to consider different settings.
However, to keep the analysis simple, in this paper, we focus on the augmented inverse probability weighting 
estimator of the CATEF.

The CATEF is mathematically equivalent to ``$V$-adjusted variable importance'' of \cite{vanderLaan06}, who proposes it as a measure of variable importance in prediction. 
\cite{vanderLaan06} proposes a doubly robust estimator of $V$-adjusted variable importance.
Contrary to ours, he considers the projection of the $V$-adjusted variable importance
 on a parametric working model and 
does not consider a nonparametric estimation.
Moreover, a uniform confidence band is not examined in \cite{vanderLaan06}.

In a recent paper, \cite{Abrevaya-et-al-JBES} consider the estimation of the CATEF\footnote{Our paper is independent of \cite{Abrevaya-et-al-JBES} and it is started without knowing their work.};  however, there are two main differences of this paper relative to \cite{Abrevaya-et-al-JBES}.
First, we propose the doubly robust procedure to estimate the CATEF.
\cite{Abrevaya-et-al-JBES} consider the inverse probability weighting estimator.
The inverse probability weighting estimator suffers from model misspecification when the propensity score model is misspecified
and from the curse of dimensionality when it is estimated nonparametrically. 
Second, we present a method to construct a uniform confidence band, whereas
\cite{Abrevaya-et-al-JBES} only provide a pointwise confidence interval.

The remainder of the paper is organized as follows. 
Section \ref{sec:method} presents the doubly robust estimation method,
Section \ref{sec:simple:case} gives an informal description of how to construct a two-sided, symmetric uniform confidence band when the dimension of $\mathbf{X}$ is one, 
and Section \ref{sec:general:case} deals with a general case and provides formal theoretical results. In Section \ref{sec:MC}, the results of Monte Carlo simulations demonstrate that in finite samples, our doubly robust estimator works well, and the proposed confidence band has desirable coverage properties. 
Section \ref{sec:EA} gives an empirical application, and Section \ref{sec:C} concludes.
The proofs are contained in Appendix \ref{sec:A}.


\section{Doubly Robust Estimation of the Average Treatment Effect Conditional on 
 Covariates of Interest}\label{sec:method}

In this section, a doubly robust method for estimating the CATEF is proposed.
We first estimate the CATEF for all the covariates using a doubly robust method. 
We then obtain the CATEF for the covariates of interest using a nonparametric approach.

Define:
\begin{align*}
\pi (\mathbf{z}) & \equiv \mathbb{E} \left[ D|\mathbf{Z}=\mathbf{z}\right],  \\
\mu _{j}(\mathbf{z}) & \equiv  \mathbb{E} \left[ Y|\mathbf{Z}=\mathbf{z},\ D=j\right] \ \text{for }j=0,1,
\end{align*}%
where $\pi (\mathbf{z})$ is the propensity score and $\mu_j (\mathbf{z})$ for $j=0,1$ are called regression functions. Note that $\mu_j (\mathbf{z}) = E( Y_j | \mathbf{Z}= \mathbf{z})$ for $j=0,1$ under unconfoundedness.
Let $\pi (\mathbf{z},\beta )\ $and $\mu _{j}(\mathbf{z},\alpha _{j})$ for $j=0,1$ denote
parametric models of $\pi (\mathbf{z})$ and $\mu _{j}(\mathbf{z})$,
respectively.\footnote{$\mu _{j}(\mathbf{z},\alpha _{j})$ may also be called ``marginal structural models'' of \cite{Robins00}.}
A doubly robust procedure requires that either $\pi (\mathbf{z})$ or $\mu _{j}(\mathbf{z})$ for $j=0,1$ 
should be correctly specified, thereby allowing for misspecification in $\pi (\mathbf{z})$ or in $\mu _{j}(\mathbf{z})$. 
Let $\theta_0 \equiv ( \alpha_{10}^\top, \alpha_{00}^\top,\beta_0^\top)^\top$ denote the vector of true or pseudo-true parameter values that optimize some criterion functions.

We consider the augmented inverse probability weighting approach.
Let:%
\begin{eqnarray*}
\psi _{1}(\mathbf{W},\alpha _{1},\beta ) & \equiv&\frac{DY}{\pi (\mathbf{Z},\beta )}-\frac{D-\pi
(\mathbf{Z},\beta )}{\pi (\mathbf{Z},\beta )}\mu _{1}(\mathbf{Z},\alpha _{1}), \\
\psi _{0}(\mathbf{W},\alpha _{0},\beta ) & \equiv&\frac{\left( 1-D\right) Y}{1-\pi (\mathbf{Z},\beta
)}+\frac{D-\pi (\mathbf{Z},\beta )}{1-\pi (\mathbf{Z},\beta )}\mu _{0}(\mathbf{Z},\alpha _{0}), \\
\psi (\mathbf{W},\theta ) & \equiv&\psi _{1}(\mathbf{W},\alpha _{1},\beta )-\psi _{0}(\mathbf{W},\alpha
_{0},\beta ),
\end{eqnarray*}%
where $\mathbf{W} \equiv (Y,\mathbf{Z}^{\top })^{\top }\ $and $\theta \equiv (\alpha _{1}^{\top
},\alpha _{0}^{\top },\beta ^{\top })^{\top }.$ 
The first terms in $\psi_{1}(\mathbf{W},\alpha _{1},\beta ) $ and $\psi_{0}(\mathbf{W},\alpha _{0},\beta ) $ 
correspond to inverse probability weighting. 
The second terms are augmented terms that make the procedure doubly robust.

The following lemma gives 
regularity conditions under which
$g(\mathbf{x})$ is identified. 

\begin{lem}[Identification of the CATEF]\label{iden-lem}
Assume that \eqref{unconfound} holds and $0 < \pi(\mathbf{Z},\beta_0 ) < 1$ almost surely. 
Suppose that either $\beta_0$ satisfies $\mathbb{E}\left[D|\mathbf{Z}\right] = \pi
(\mathbf{Z},\beta_0 )$ almost surely or $\alpha_{10}$ and $\alpha_{00}$ satisfy
$\mathbb{E}\left[Y_1|\mathbf{Z}\right] = \mu _{1}(\mathbf{Z},\alpha _{10})$ 
and $\mathbb{E}\left[Y_0|\mathbf{Z}\right] = \mu _{0}(\mathbf{Z},\alpha _{00})$ almost surely. 
Then: 
\begin{align*}
g(\mathbf{x}) = \mathbb{E}\left[ \psi (\mathbf{W},\theta_0 )|\mathbf{X}=\mathbf{x} \right].\end{align*}
\end{lem}

Lemma \ref{iden-lem} suggests that one may estimate 
$g(\mathbf{x})$ by running the nonparametric regression of $\psi (\mathbf{W}, \hat \theta)$
on $\mathbf{X}_i$, where $\hat \theta$ is a consistent parametric estimator of $\theta_0$.
Moreover, this lemma implies that the CATEF can be identified through $\psi (\mathbf{W}, \theta_0)$ if either the regression models ($\mu_1 (\mathbf{z}, \alpha_1)$ and $\mu_0 (\mathbf{z}, \alpha_0)$) or the propensity score model ($\pi (\mathbf{z},\beta)$) is correctly specified (or both). 
That is, even if $\mu_1 (\mathbf{z}, \alpha_1)$ and $\mu_0 (\mathbf{z}, \alpha_0)$ do not represent the true conditional expectation functions, provided that $\pi (\mathbf{z}, \beta)$ is correct, the CATEF is identified. Similarly, even if $\pi (\mathbf{z}, \beta)$ is misspecified, 
provided that $\mu_1 (\mathbf{z}, \alpha_1)$ and $\mu_0 (\mathbf{z}, \alpha_0)$ are correct, 
the CATEF is identified.

\begin{rem}
In this paper, we focus on cases in which $\mathbf{X}$ is continuous. When $\mathbf{X}$ is discrete, the CATEF can be estimated by the sample average of $\psi (\mathbf{W}, \hat \theta)$ using the sub-sample for each possible value of $\mathbf{X}$ and an estimator $\hat \theta$ of $\theta_0$. Moreover, constructing a confidence band is standard when $\mathbf{X}$ takes a finite number of values.
\end{rem}

\subsection{Parametric Estimation of $\theta$}

For concreteness, we consider the following estimation procedure for $\theta_0$.
However, how $\theta_0$ is estimated does not alter our results provided that the rate of convergence is sufficiently fast so that Assumption \ref{main-assumption}(7) given below is satisfied.
For each $j = 0,1$, we estimate $\alpha_j$ by least squares:
\begin{align}\label{ols-est}
\hat  \alpha_j   \equiv \argmin_{\alpha_j}  
\sum_{i=1}^n D_i^j (1-D_i)^{1-j}  [Y_i - \mu_j (\mathbf{Z}_i, \alpha_j)]^2.
\end{align}
We estimate $\beta$ by maximum likelihood (e.g., probit or logit): 
\begin{align}
 \hat  \beta   \equiv \argmax_{\beta}  \sum_{i=1}^n \left(
 D_i \log \pi(\mathbf{Z}_i, \beta) + (1-D_i) \log (1- \pi (\mathbf{Z}_i,\beta))\right).
\end{align}

\begin{rem}
When the dimension of $\mathbf{Z}$ is not too high, an alternative to parametric estimation of 
$\psi (\mathbf{W},\theta_0 )$ is to estimate its nonparametric counterpart via local polynomial estimators
as in \cite{Rothe/Firpo:16}. However, this would not work when the   dimension of $\mathbf{Z}$ is sufficiently high (see related remarks in \cite{Rothe/Firpo:16}).
The latter is the case we focus on in the paper.
\end{rem}

\subsection{Local Linear Estimation of $g$}
We consider a
local linear estimator of $g(\mathbf{x})$.
 Assume that $g(\mathbf{x})$ is twice
continuously differentiable. 
For each $\mathbf{x}=(x_{1},\ldots ,x_{d})$, the local linear estimator of $g(\mathbf{x})$ can be obtained by
minimizing: 
\begin{equation*}
S_{n}(\gamma )\equiv \sum_{i=1}^{n}\left[ \psi (\mathbf{W}_{i},\hat \theta
)- \gamma_0 - \gamma_1 ^{\top }\left(\mathbf{X}_i - \mathbf{x} \right) \right]^2 
\mathbf{K} \left( \frac{\mathbf{X}_i - \mathbf{x}}{h_n} \right)
\end{equation*}%
with respect to $\gamma \equiv (\gamma_0, \gamma_1^\top) ^\top \in \mathbb{R}^{d+1}$, where $\mathbf{K}(\cdot)$ is a kernel function on $\mathbb{R}^d$
and $h_n$ is a sequence of bandwidths. More specifically, let $%
\hat{g}( \mathbf{x})=\mathbf{e}_{1}^{\top }\hat{\gamma}(\mathbf{x})$, where $%
\hat{\gamma}(\mathbf{x})\equiv \arg \min_{\gamma \in \mathbb{R}%
^{d+1}}S_{n}(\gamma )$ and $\mathbf{e}_{1}$ is a column vector
whose first entry is one, and the rest are zero.

\subsection{Effect of First Stage Estimation}

In our setting, we can carry out inference as if $\theta_0$ were known.
This result would not be a surprise given that our first-stage estimation is parametric and our second-stage estimation is nonparametric: the rate of the convergence in the first-stage estimation is faster than that of the second stage. 
This feature of no first-order effect of the first-stage estimation in the second stage turns out to be more general than our setup.
It is indeed closely related to doubly robustness.

If we model $g(\mathbf{x})$ parametrically or more generally approximate $g(\mathbf{x})$ by linear projection, it can be estimated by running an OLS of  $\psi (\mathbf{W},\hat \theta )$ on $\mathbf{X}$.
Because of the built-in feature of double robustness, it can be shown that the limiting distribution of 
 the OLS estimator of $\psi (\mathbf{W},\hat \theta )$ on $\mathbf{X}$ is equivalent to that of the infeasible 
OLS estimator of $\psi (\mathbf{W}, \theta_0 )$ on $\mathbf{X}$.
Furthermore, even if we estimate $\pi(\cdot)$ and $\mu_j(\cdot)$ $(j=0,1)$ nonparametrically when the dimension of $\mathbf{Z}$ is moderate, there will be no estimation effect from the first stage as well.
For example, see \cite{chen2008}, \cite{Rothe/Firpo:16} and \cite{CEIN2016} among others for related results.

\section{An Informal Description of a Uniform Confidence Band}\label{sec:simple:case}

In this section, we provide an informal description of how to construct a two-sided, symmetric uniform confidence band. 
For simplicity, we focus on the leading case where $d=1$.
Let $\mathcal{I} \equiv [a,b]$ denote an interval of interest for which we build a uniform confidence band. Assume that $\mathcal{I}$ is a subset of the support of $X$.
We use nonbold $x$ to mean that $x$ is one-dimensional.

\begin{algorithm}
Carry out the following steps to construct a $(1-\alpha)$ uniform confidence band.
\begin{enumerate}
\item Obtain $\hat{g}(x)$ using a local linear estimator with 
a bandwidth $h_n$ such that: 
\begin{align*}
h_n = \widehat{h} \times  n^{1/5} \times n^{-2/7},
\end{align*}
where $\widehat{h}$ is a commonly used optimal bandwidth in the literature   
(for example, the plug-in method of \cite{Ruppert:Sheather:Wand:1995} which is explained in Appendix \ref{ap-bandwidth}).
We use the Gaussian kernel in our simulations and empirical application.

\item  Obtain the pointwise standard error $\hat{s}(x)/(n h_n)^{1/2}$ of $\hat{g}(x)$ 
by constructing a feasible version of the asymptotic standard error formula:
\begin{align}\label{se-form}
\frac{\hat s(x)}{(nh_n)^{1/2}} \equiv \left\{ [n h_n \hat f_X(x)]^{-1} \int K^2(u) du \, \hat \sigma^2(x) \right\}^{1/2},
\end{align}
where $\hat{f}_X$ is the kernel 
density estimator:
\begin{align*}
 \hat{f}_X (x) = \frac{1}{nh_n} \sum_{i=1}^n K \left(\frac{X_i - x}{h_n} \right),
\end{align*}
and $\hat{\sigma}^2(x)$ is the conditional variance function estimator:
\begin{align}\label{sigma-x-est-dim1}
 \hat{\sigma}^2 (x)
= \frac{1}{(n-\dim (\theta) ) h_n} \sum_{i=1}^N \frac{\hat U_i^2}{\hat{f}_X (x)} K \left( \frac{X_i - x}{h_n}\right).
\end{align}
Here, $\hat U_i = \psi ( \mathbf{W}_i, \hat \theta) - \hat g(X_i)$
and $\dim(\theta)$ is the dimension of $\theta$. 



\item To compute a critical value $c(1-\alpha)$, define:
\begin{equation*}
\lambda \equiv \frac{ - \int K(u) K^{\prime \prime }(u) du }{\int K^2(u) du }.
\end{equation*}
Note that $\lambda =0.5$ if $K (\cdot)$ is the Gaussian kernel.\footnote{Note that $\lambda= 1.98$ for the biweight kernel and $\lambda= 2.5$ for the Epanechnikov kernel.}
 Let:
\begin{equation*}
a_n \equiv a_{n}(\mathcal{I})= \left( 2\log (h_{n}^{-1}(b-a))+2\log \frac{\lambda ^{1/2}}{2\pi }%
\right) ^{1/2}.
\end{equation*}
Now set the critical value for the two-sided symmetric uniform confidence band by: 
\begin{equation*}
c(1-\alpha) \equiv  \left( a_n^2 - 2  \log \{ \log [(1-\alpha )^{-1/2}] \} \right)^{1/2}.
\end{equation*}

\item For each $x \in \mathcal{I}$, we set the two-sided symmetric confidence band:
\begin{align*}
\hat{g}(x) - c(1-\alpha) \frac{\hat{s}(x)}{\sqrt{nh_n}}
\leq
g(x) 
\leq 
\hat{g}(x) + c(1-\alpha) \frac{\hat{s}(x)}{\sqrt{nh_n}}.
\end{align*}
\end{enumerate}
\end{algorithm}

\bigskip

We make some remarks on the proposed algorithm.
In step (1), 
the factor $n^{1/5} \times n^{-2/7}$ is multiplied in the definition of $h_n$ to ensure that the bias is asymptotically negligible by undersmoothing.
In step (2), one can estimate $f_X$ and $\sigma^2(x) \equiv \text{Var}\left[ \psi (\mathbf{W},\theta_0 )|X=x \right]$ using the standard kernel density and regression estimators with the same kernel function $K(\cdot)$ and the same bandwidth $h$ and also with an estimator of $\theta_0$.
In step (3), we may restrict the bandwidth such that $h_{n} \ll (b-a)$ (which is satisfied asymptotically), thereby imposing the condition that $\log (h_{n}^{-1}(b-a))$ is positive.
The critical value proposed in step (3) is strictly positive if $\alpha$ is not too close to one or if $n$ is large enough.

\begin{rem}
It is straightforward to modify the algorithm above to construct one-sided symmetric confidence bands.
Define a new critical value by
\begin{equation*}
c_{\text{one-sided}}(1-\alpha) \equiv  \left( a_n^2 - 2  \log \{ \log [(1-\alpha )^{-1}] \} \right)^{1/2}.
\end{equation*}
Then, for each $x \in \mathcal{I}$, we set the one-sided symmetric confidence bands:
\begin{align*}
\hat{g}(x) - c_{\text{one-sided}}(1-\alpha) \frac{\hat{s}(x)}{\sqrt{nh_n}}
\leq g(x), 
\end{align*}
or 
\begin{align*}
g(x) 
\leq 
\hat{g}(x) + c_{\text{one-sided}}(1-\alpha) \frac{\hat{s}(x)}{\sqrt{nh_n}}.
\end{align*}
\end{rem}

\begin{rem}{\label{rem-algorithm-multivariate}}
 When $x$ is more than one dimension, the algorithm may be revised as follows. 
Obviously, we need to use multivariate kernels and pointwise standard errors should be adjusted because the rate of convergence becomes $nh_n^d$.
The value of $\lambda$ stays the same when we use a product kernel. For example, if $\mathbf{K}$ is the product Gaussian kernel, then $\lambda =0.5$.
The formulas of $a_n$ and $c(1-\alpha)$ need to be changed.
$a_n$ is the largest solution to the following equation:
\begin{align*}
\text{mes}(\mathcal{I}) {h_n}^{-d} \lambda^{d/2} (2\pi)^{-(d+1)/2 }
a_n^{d-1} \exp(-a_n^2/2) = 1,
\end{align*}
where $\text{mes}(\mathcal{I})$ is the Lebesgue measure of $\mathcal{I}$.
When $d =2$, the critical value has the form $c(1-\alpha) \equiv a_n + c/a_n$, where $c$ is the smallest value that satisfies
\begin{align*}
 \exp \left( -2 e^{-c -c^2 / 2 a_n^2}\right)\left( 1+ \frac{c}{a_n^2}\right) \ge 1-\alpha.
\end{align*}
When $d=3$, we have that $c(1-\alpha) \equiv a_n + c/a_n$, where $c$ is the smallest value that satisfies
\begin{align*}
 \exp \left( -2 e^{-c -c^2 / 2 a_n^2}\right)\left( \left( 1+ \frac{c}{a_n^2}\right)^2 - 2 \frac{1}{a_n^2} \right) \ge 1- \alpha.
\end{align*}
We note that in this paper, we assume that $d <4$ (see Assumption \ref{main-assumption}).
\end{rem}

\begin{rem}\label{rem-gumbel}
We may compare our proposal with the critical value based on 
the $(1-\alpha)$ quantile of the Gumbel distribution, which is given by:
\begin{equation*}
c_{\infty}(1-\alpha) \equiv a_{n} + \frac{-\log \{ \log [(1-\alpha)^{-1/2}] \}}{a_{n}}.
\end{equation*}
Note that: 
$$
c_{\infty}(1-\alpha) - c(1-\alpha) = \left[\frac{-\log \{ \log [(1-\alpha)^{-1/2}] \}}{a_{n}}\right]^2,
$$
which is strictly positive for small $\alpha$ but converges to zero as $a_n$ diverges.
Hence, we expect that in finite samples, the confidence band based on $c_{\infty}(1-\alpha)$ is too wide and 
has a higher coverage probability than the nominal level. 
It is shown in the next section that the critical value based on the Gumbel distribution is accurate only up to the logarithmic rate, where our proposed critical value is precise in a polynomial rate. 
This is because our proposal uses a higher-order expansion of \cite{Piterbarg:96},
whose approximation error is of a polynomial rate. 
See Theorem \ref{main-uniform} in Section \ref{sec:general:case} for details.
\end{rem}

\begin{rem}
Our construction of critical values is based on a simple analytic method that is easy to compute. 
Alternatively, one may rely on bootstrap methods to compute critical values for the uniform confidence band.
For example, see \cite{claeskens2003}
for smoothed bootstrap confidence
bands and \cite{CLR} 
 for multiplier bootstrap confidence bands. 
\cite{chernozhukov2013}  
 show that in general settings including high dimensional models, 
Gaussian multiplier bootstrap methods yield critical values for which the approximation error decreases polynomially in the sample size. Roughly speaking, both our simple analytic correction and multiplier bootstrap methods
yield critical values that are accurate at polynomial rates.
A refined theoretical analysis is necessary to determine which type of the critical value is better asymptotically.
\end{rem}

\begin{rem}
The proposed confidence band can be used to test whether the CATEF is constant.
Suppose that our null hypothesis is that $g(\mathbf{x})$ is constant in $ \mathcal{I}$.
This null hypothesis can be written as $g(\mathbf{x}) = g_{\mathcal{I}}$, where $g_{\mathcal{I}}
= \mathbb{E}[ g(\mathbf{x}) | \mathbf{x} \in \mathcal{I}]$.
Since $g_{\mathcal{I}}$ can be estimated at the parametric ($\sqrt{n}\,$) rate and the estimator thus converges faster than $\hat g(\mathbf{x})$, we can ignore the estimation error for $g_{\mathcal{I}}$.
We reject the constancy of $g(\mathbf{x})$, if the confidence band does not include the estimate of $g_{\mathcal{I}}$ for some $\mathbf{x} \in \mathcal{I}$.

\end{rem}

\section{Asymptotic Theory}\label{sec:general:case}

In this section, we establish asymptotic theory.
Let $U \equiv \psi (\mathbf{W},\theta_0 ) - g(\mathbf{X})$ and let 
 $U_i \equiv \psi (\mathbf{W}_i,\theta_0 ) - g(\mathbf{X}_i)$ for $i=1,\ldots,n$. 
Let $\hat s^2 (\mathbf{x})$ be the estimator of the asymptotic variance of $\hat g(\mathbf{x})$.
Let $s_n^2(\mathbf{x})$ denote the population version of the asymptotic variance of the estimator: 
 \begin{align*}
s_n^2(\mathbf{x}) 
\equiv   \frac{1}{h_n^{d}} \mathbb{E}\left[ \frac{U^2}{f^2_\mathbf{X}(\mathbf{x})} \mathbf{K}^2 \left( \frac{\mathbf{X} - \mathbf{x}}{h_n} \right) \right].
\end{align*}
Assume that the $d$-dimensional kernel function is the product of $d$ univariate kernel functions.\
That is, $\mathbf{K}(\mathbf{s}) = \prod_{j=1}^d K(s_j)$, where
$\mathbf{s} \equiv (s_1,\ldots,s_d)$ is a $d$-dimensional vector 
and $K(\cdot)$ is a kernel function on $\mathbb{R}$.
Let $\rho_d (\mathbf{s}) = \prod_{j=1}^d \rho(s_j)$, where:
\begin{align}\label{kl-rho}
\rho(s_j)  \equiv \frac{\int K(u) K(u-s_j) du}{\int K^2(u) du },
\end{align}
for each $j$.
We make the following assumptions.

 \begin{assum}\label{main-assumption}
 Let $d < 4$.
\begin{enumerate}
\item $\mathcal{I} \equiv \prod_{j=1}^d [a_j,b_j]$, where $a_j < b_j, j=1,\ldots,d$, and $\mathcal{I}$ is a strict subset of the support
of $\bm{X}$. 
\item The distribution of $\mathbf{X}$ has a bounded Lebesgue density $f_{\mathbf{X}}(\cdot)$ on $\mathbb{R}^d$.
Furthermore, $f_{\mathbf{X}}(\cdot)$ is bounded below from zero with continuous derivatives on $\mathcal{I}$.
\item The density of $U$ is bounded,  $\mathbb{E}[U^2|\mathbf{X}=\mathbf{x}]$ is continuous on $\mathcal{I}$, and $\sup_{\mathbf{x} \in \mathbb{R}^d} \mathbb{E}[U^4|\mathbf{X}=\mathbf{x}] < \infty$.
\item $g(\cdot)$ is twice continuously differentiable on $\mathcal{I}$.
\item $\mathbf{K}(\mathbf{s}) = \prod_{j=1}^d K(s_j)$, where
 $K(\cdot)$ is a kernel function on $\mathbb{R}$ that has finite support on $[-1,1]$, $\int_{-1}^1 u K(u) du=0$, $\int_{-1}^1  K(u) du=1$,
symmetric around zero, and six times differentiable.
\item $h_n = C n^{-\eta}$, where $C$ and $\eta$ are positive constants such that $\eta < 1/(2d)$ and $\eta > 1/(d+4)$. 
\item $\inf_{n \geq 1} \inf_{\mathbf{x} \in \mathcal{I}} s_n(\mathbf{x}) > 0$ and $s_n(\mathbf{x})$ is continuous for each $n \geq 1$. 
Furthermore, $\mathbf{x} \mapsto \mathbb{E}\left[ U^2 | \mathbf{X} =\mathbf{x} \right]f_{\mathbf{X}}(\mathbf{x})$ is Lipschitz continuous.
\item There exists an estimator $\hat s^2 (\mathbf{x})$ such that
\begin{align*}
\sup_{\mathbf{x} \in \mathcal{I}} \left| \hat s^2 (\mathbf{x}) - s_n^2(\mathbf{x}) \right| = O_p( n^{-c})
\end{align*}
for some constant $c > 0$.

\item $\max\left\{  (n h_n^d)^{1/2}|\psi (\mathbf{W}_{i},\hat \theta) - \psi (\mathbf{W}_{i},\theta_0) |: i=1,\dots,n \right\} = O_p( n^{-c})$ for some constant $c > 0$.
\end{enumerate}
\end{assum}

Most of the assumptions are standard. 
Condition (2) of Assumption \ref{main-assumption} rules out discrete covariates. If all
regressors are discrete, then the  estimation problem reduces to a parametric estimation problem. 
In this case, one may consider a multiple testing approach as in \cite{Lee:Shaikh:2014} by defining subpopulations with observed cells of discrete covariates. If some covariates are discrete and  others are continuous, then one may use a smoothing approach proposed in \cite{Li:Racine:2004}. 

Condition (5) assumes that the kernel function has finite support. This assumption is for 
the simplicity of the paper and can be dropped at the expense of complicated proofs.
It also assumes that the kernel function is differentiable. 
This assumption is crucial and excludes, for example, the uniform kernel.
One of the bandwidth conditions in $h_n$ (that is, $\eta > 1/(d+4)$ in condition (6)) imposes undersmoothing, so that we can ignore the bias asymptotically. 
The rule-of-thumb bandwidth proposed in Section \ref{sec:simple:case} satisfies the required rate conditions.

\begin{rem}
Note that $d < 4$ is necessary to ensure that  $\eta < 1/(2d)$ and $\eta > 1/(d+4)$ can hold jointly.
It is possible to extend our asymptotic theory to the case that $d \geq 4$ using a higher-order local polynomial estimator under stronger smoothness conditions.
In this paper, we limit our attention to the local linear estimator since we are mainly interested in  low dimensional $\bm{x}$.
\end{rem}

\begin{rem}
An estimator of $\hat s^2 (\mathbf{x})$ is readily available. For example, we may consider
 \begin{align}\label{sigma-x-est-dim-d}
\hat s^2(\mathbf{x}) 
=
  \frac{\hat{\sigma}^2(\mathbf{x})}{ \hat f_\mathbf{X}(\mathbf{x})} 
\int \mathbf{K}^2 \left( \mathbf{u} \right) d \mathbf{u},
\end{align}
where $\hat{f}_\mathbf{X}(\cdot)$ is the kernel density estimator
and 
$\hat{\sigma}^2(\mathbf{x})$ is a nonparametric estimator  of 
${\sigma}^2(\mathbf{x})$  using
$\{ (\hat{U}_i^2, \mathbf{X}_i): i=1,\ldots,n \}$ with 
$\hat{U}_i \equiv \psi (\mathbf{W}_i,\hat{\theta} ) - \hat{g}(\mathbf{X}_i)$.
Recall \eqref{sigma-x-est-dim1} for its concrete form for the one-dimensional case.
Alternatively, we may set
 \begin{align*}
\hat s^2(\mathbf{x}) 
=   \frac{1}{nh_n^{d}} \sum_{i=1}^n  \frac{\hat{U}_i^2}{\hat{f}^2_\mathbf{X}(\mathbf{x})} \mathbf{K}^2 \left( \frac{\mathbf{X} - \mathbf{x}}{h_n} \right).
\end{align*}
For either estimator, it is straightforward to verify condition (8) of Assumption \ref{main-assumption} using the standard results in kernel estimation.
\end{rem}

\begin{rem}
Note that Condition (9) of Assumption \ref{main-assumption} is merely a sufficient (but not necessary) condition.
This condition is satisfied, for example, if
$\| \hat \theta - \theta_0 \| = O_p( n^{-1/2})$,
functions $\beta \mapsto \pi (\mathbf{Z},\beta)$
and
$\alpha_j \mapsto \mu _{j}(\mathbf{Z},\alpha _{j}), j =0,1,$ are
Lipschitz continuous,
$\pi (\mathbf{Z},\beta_0 )$ is bounded between $\epsilon$ and $1-\epsilon$ for some 
constant $\epsilon > 0$, provided that 
some weak moment conditions on $(Y, \mathbf{Z})$ hold. 
 \end{rem}

Let $a_n \equiv a_n(\mathcal{I})$ be the largest solution to the following equation:
\begin{align}  \label{b1}
\text{mes}(\mathcal{I}) {h_n}^{-d} \lambda^{d/2} (2\pi)^{-(d+1)/2 }
a_n^{d-1} \exp(-a_n^2/2) = 1,
\end{align}
where $\text{mes}(\mathcal{I})$ is the Lebesgue measure of $\mathcal{I}$; that is, 
$\text{mes}(\mathcal{I}) = \prod_{j=1}^d (b_j-a_j)$ and: 
\begin{equation}\label{def-lambda}
\lambda = \frac{ - \int K(u) K^{\prime \prime }(u) du }{\int K^2(u) du }.
\end{equation}
The following is the main theoretical result of our paper.

\begin{thm}\label{main-uniform}
Let Assumption \ref{main-assumption} hold. Then there exists $\kappa > 0$ such that, 
uniformly in $t$, on any finite  interval:
\begin{align}\label{main-approx}
\begin{split}
&\mathbb{P}\left( a_n  
\left[  \max_{\mathbf{x} \in \mathcal{I}} \left| \frac{\hat{g}(\mathbf{x}) -  g(\mathbf{x})}{\hat{s}(\mathbf{x})} \right|  - a_n  \right]  < t \right) = \\
&\exp \left( - 2e^{-t - t^2/2a_n^2} \right)
\sum_{m=0}^{\lfloor (d-1)/2 \rfloor} h_{m,d-1} a_n^{-2m}
\left( 1 + \frac{t}{a_n^2} \right)^{d-2m-1}
+ O( n^{-\kappa}),
\end{split}
\end{align}
as $n \rightarrow \infty$, where $h_{m,d-1}  \equiv \frac{(-1)^m (d-1)!}{ m! 2^m (d-2m-1)!}$
and $\lfloor \cdot \rfloor$ is the integer part of a number.  
\end{thm}

Notice that the approximation error is of a polynomial rate. 
As a result, a critical value based on the leading term of the right-hand side of
\eqref{main-approx} provides a better approximation than one based on the Gumbel approximation.
The result in Theorem \ref{main-uniform} may be of independent interest for constructing the uniform confidence band in nonparametric regression beyond the scope of estimating the CATEF in our context.

\begin{rem}
In a setting different from here,  \cite{Lee/Linton/Whang:06} propose analytic critical values based on \cite{Piterbarg:96} in order to test for stochastic monotonicity, compare its performance with the bootstrap critical values in their Monte Carlo experiments, and find that both perform well in finite samples.
However, the discussions in \cite{Lee/Linton/Whang:06} are informal and rely on the results of Monte Carlo experiments without the formal proof of establishing the polynomial approximation error. 
\end{rem}

The conclusion of the theorem can be simplified for special cases. 
In particular, if $d=1$, then:
\begin{align*}
\mathbb{P}\left( a_n  
\left[  \max_{\mathbf{x} \in \mathcal{I}} \left| \frac{\hat{g}(\mathbf{x}) -  g(\mathbf{x})}{\hat{s}(\mathbf{x})} \right|  - a_n  \right]  < t \right) = \exp \left( - 2e^{-t - t^2/2a_n^2} \right)
+ O( n^{-\kappa}),
\end{align*}
where $a_n$ is the largest solution to  
$\text{mes}(\mathcal{I}) {h_n}^{-1} \lambda^{1/2} (2\pi)^{-1 }\exp(-a_n^2/2) = 1$. 
Also, if $d=2$, then:
\begin{align*}
\mathbb{P}\left( a_n  
\left[  \max_{\mathbf{x} \in \mathcal{I}} \left| \frac{\hat{g}(\mathbf{x}) -  g(\mathbf{x})}{\hat{s}(\mathbf{x})} \right|  - a_n  \right]  < t \right) = \exp \left( - 2e^{-t - t^2/2a_n^2} \right)\left( 1 + \frac{t}{a_n^2} \right)
+ O( n^{-\kappa}),
\end{align*}
where $a_n$ is the largest solution to  
 $\text{mes}(\mathcal{I}) {h_n}^{-2} \lambda^{2} (2\pi)^{-3/2 }
a_n \exp(-a_n^2/2) = 1$. 

\begin{rem}
It is standard to obtain pointwise confidence intervals based on normal approximations. 
Recall that our two-sided symmetric uniform confidence band has the form: 
\begin{align}\label{pointwise-ci}
\hat{g}(\mathbf{x}) - c(1-\alpha) \frac{\hat{s}(\mathbf{x})}{\sqrt{nh_n^d}}
\leq
g(x) 
\leq 
\hat{g}(\mathbf{x}) + c(1-\alpha) \frac{\hat{s}(\mathbf{x})}{\sqrt{nh_n^d}},
\end{align}
where $c(1-\alpha)$ is obtained from Theorem \ref{main-uniform}.
To obtain    two-sided symmetric pointwise confidence intervals, we just need to replace 
$c(1-\alpha)$ with the usual normal critical value $\Phi^{-1}(1-\alpha/2)$, where $\Phi(\cdot)$ is the standard normal cumulative distribution function. 
The pointwise confidence interval given in \eqref{pointwise-ci} is different from the one resulting from 
\citet[Theorem 2]{Abrevaya-et-al-JBES} in terms of the formula for $\hat{s}(\mathbf{x})$ in \eqref{sigma-x-est-dim-d}. In their case, they need to estimate 
$\hat{\sigma}^2(\mathbf{x})$   using
$\{ (\tilde{U}_i^2, \mathbf{X}_i): i=1,\ldots,n \}$ with 
\begin{align*}
\tilde{U}_i \equiv  \frac{D_i Y_i}{\pi (\mathbf{Z}_i,\hat{\beta} )}- \frac{\left( 1-D_i\right) Y_i}{1-\pi (\mathbf{Z}_i,\hat{\beta})}
\hat{g}(\mathbf{X}_i).
\end{align*}%
\end{rem}

\begin{rem}
 A one-sided version of the uniform confidence band is readily available. Combining Theorems 14.1 and 14.2 of \cite{Piterbarg:96}
with arguments identical to those used in the proof of Theorem \ref{main-approx-one-sided} yields the following proposition. Under Assumption \ref{main-assumption},  there exists $\kappa > 0$ such that, 
uniformly in $t$, on any finite  interval:
\begin{align}\label{main-approx-one-sided}
\begin{split}
&\mathbb{P}\left( a_n  
\left[  \max_{\mathbf{x} \in \mathcal{I}}  \frac{\hat{g}(\mathbf{x}) -  g(\mathbf{x})}{\hat{s}(\mathbf{x})} - a_n  \right]  < t \right) = \\
&\exp \left( - e^{-t - t^2/2a_n^2} \right)
\sum_{m=0}^{\lfloor (d-1)/2 \rfloor} h_{m,d-1} a_n^{-2m}
\left( 1 + \frac{t}{a_n^2} \right)^{d-2m-1}
+ O( n^{-\kappa}),
\end{split}
\end{align}
as $n \rightarrow \infty$.  
Note that the only differences between \eqref{main-approx} and \eqref{main-approx-one-sided} are that 
(i) there is no absolute value on the left side of the equation in \eqref{main-approx-one-sided}
and (ii) there is no factor 2 inside the exponential function in \eqref{main-approx-one-sided}. 
Hence, for example, if $d=1$, then:
\begin{align*}
\mathbb{P}\left( a_n  
\left[  \max_{\mathbf{x} \in \mathcal{I}}  \frac{\hat{g}(\mathbf{x}) -  g(\mathbf{x})}{\hat{s}(\mathbf{x})}   - a_n  \right]  < t \right) = \exp \left( - e^{-t - t^2/2a_n^2} \right)
+ O( n^{-\kappa}).
\end{align*} 
\end{rem}

\subsection{Construction of critical values}

We use the leading term on the right-hand side of \eqref{main-approx}
as a distribution-like function to construct a uniform confidence band. 
For example, if $d=1$, we may construct a critical value $c(1-\alpha)$ that satisfies:
\begin{align*}
F_{n,1}(c) \geq 1 - \alpha,
\end{align*}
where $F_{n,1}(t) \equiv \exp \left( - 2e^{-t - t^2/2a_n^2} \right)$.
This yields the critical value presented in the Algorithm of Section \ref{sec:simple:case}.
Similarly, if $d=2$, we can use: 
\begin{align*}
F_{n,2}(c) \geq 1 - \alpha,
\end{align*}
where $F_{n,2}(t) \equiv \exp \left( - 2e^{-t - t^2/2a_n^2} \right)\left( 1 + \frac{t}{a_n^2} \right)$.
In finite samples, it might be useful to impose monotonicity of $F_{n,j}(\cdot)$
by rearrangement 
(see, e.g., \cite{CF-VG:09}).

\begin{rem}
Theorem \ref{main-uniform} implies that: 
\begin{align*}
&\lim_{n \rightarrow \infty} \mathbb{P}\left( a_n  
\left[  \max_{\mathbf{x} \in \mathcal{I}} \left| \frac{\hat{g}(\mathbf{x}) -  g(\mathbf{x})}{\hat{s}(\mathbf{x})} \right| - a_n  \right]  < t \right) = \exp \left( - 2e^{-t } \right).
\end{align*}
Thus, one may construct analytical critical values based on the Gumbel distribution. However, this approximation is accurate only up to the logarithmic rate in view of Theorem \ref{main-uniform}. 
\end{rem}

\section{Monte Carlo Experiments}\label{sec:MC}

In this section, we present the results of Monte Carlo experiments.
These experiments are conducted to see the finite sample performances of the proposed doubly robust estimator and the proposed uniform confidence band.
The simulations are conducted by R 3.3.1 with Windows 10.
The number of replications is $5000$.

\subsection{Data generating process}

The data generating process follows the potential outcome framework.
The notations for the variables are the same as those used in the theoretical part of the paper.
We consider cases with $p=10, 30$ and $N= 500, 2000$.

The data generating process is the following.
The vector of covariates $\mathbf{Z}= (X_1, \dots, X_p)^{\top} $ is generated by:
$\mathbf{Z} \sim N(0,I_p)$, where $I_p$ is the $p$-dimensional identity matrix.
The potential outcomes are generated by:
\begin{align*}
 Y_1 = 10 + \sum_{k=1}^p \frac{1}{\sqrt{p}} X_k + v, \quad Y_0 =0,
\end{align*}
where $v \sim N(0, 1)$ and $v$ is independent of $\mathbf{Z}$.
The treatment status $D$ is generated by:
\begin{align*}
 D =\mathbf{1} \left\{ \Lambda \left(\sum_{k= p/2}^p \frac{1}{\sqrt{p/2}}X_k \right) > U\right\},
\end{align*}
where $U \sim U[0,1]$, $U$ is independent of $(\mathbf{Z}^{\top}, v)$
and $\Lambda$ is the logistic function. 
Thus, the propensity score is $\pi(\mathbf{Z}) = \Lambda \left(\sum_{k= p/2}^p X_k/ \sqrt{p/2}\right)$.
The observed outcome is $Y=DY_1$.

The parameter of interest is the CATEF for $\mathbf{X} =X_1$.
In our specification, the CATEF 
can be written as:
\begin{align*}
 CATEF (x_1) = 10 + x_1/\sqrt{p} .
\end{align*}
We examine the performance of various statistical procedures regarding this CATEF.

\subsection{Model specification}

To estimate and conduct statistical inferences on $CATEF(x_1)$ using our doubly robust procedure, we need to specify a model for the regression $\mu_j (\mathbf{z})$ for $j=0,1$ and a model for the propensity score $\pi (\mathbf{z})$.
We consider two regression models and two propensity score models.
One of two models is correctly specified, but the other model is misspecified.
We note that our doubly robust procedure is predicted to work well provided that at least one of the regression model and the propensity score model is correctly specified.

We first discuss the model specifications for the regression part.
The first regression model is: 
\begin{align*}
\mu_1 (\mathbf{z}, \alpha_1) = \alpha_{10} + \sum_{k=1}^p \alpha_{1k} X_k, \quad
\mu_0 (\mathbf{z}, \alpha_0) = \alpha_{00} + \sum_{k=1}^p \alpha_{0k} X_k .
 \end{align*}
This model is correctly specified. 
The coefficients are estimated by OLS using $(1, X_1, \dots, X_p)$ as the explanatory variable.
The second regression model, which is misspecified, is: 
\begin{align*}
\mu_1 (\mathbf{z}, \alpha_1)  = \alpha_{10} + \sum_{k=1}^{p/2} \alpha_{1k} X_k ,\quad 
\mu_0 (\mathbf{z}, \alpha_0)  = \alpha_{00} + \sum_{k=1}^{p/2} \alpha_{0k} X_k .
 \end{align*}
The model is estimated by OLS using $(1, X_1, \dots, X_{p/2})$ as explanatory variables.
This model is misspecified because it suffers from sample selection bias introduced by omitting the second half of the regressors which affects the treatment status.

We also consider two models for propensity score. 
The model for propensity score is:
\begin{align*}
 \pi (\mathbf{z}, \beta) = \Lambda \left(\beta_0 + \sum_{k=1}^p \beta_k X_k\right).
\end{align*}
The misspecified model is: 
\begin{align*}
\pi(\mathbf{z},\beta) = \Lambda \left(\beta_0 + \sum_{k=1}^{p/2} \beta_k X_k \right).
\end{align*}
Similarly to the case of the regression part, misspecification is introduced by omitting the second half of the regressors.
The models for propensity score are estimated by maximum likelihood.

We estimate $CATEF(x_1)$ for $x_1 \in \{ -1, -0.5, 0, 0.5, 1\}$ and compute the mean bias (``MEAN''), standard deviation (``SD''), the average of standard error for $\widehat{CATEF} (x_1)$ (``ASE''), and the root mean squared error (``RMSE'').
The local linear regression is conducted with the Gaussian kernel, and the preliminary bandwidth ($\hat h$ in Algorithm (1)) is chosen by the method of \cite{Ruppert:Sheather:Wand:1995}.
We also compute the ``BIAS'', ``SE'' and ``RMSE'' of the corresponding inverse probability weighting estimators and the regression adjustment estimators. 
Note that the difference between the proposed method and those alternative methods arises only in the estimation of $\psi (\mathbf{W}, \theta_0)$ and the other steps are the same.

We examine the coverage probability of the uniform confidence band for $CATEF (x_1)$ for the range $-1 \le x_1 \le 1$.
The nominal coverage probabilities that we consider are 99\%, 95\% and 90\%.
We compute the empirical coverage (``CP''), the mean critical value (``Mcri''), and the standard deviation of critical value (``Sdcri'').
We also compute the coverage probabilities of the confidence band based on the critical values computed by the Gumbel approximation (``GCP'').

\subsection{Results}

Tables \ref{tab-est-k10} and \ref{tab-est-k30} summarize the results on the properties of the estimators.
In both tables, DR refers to our doubly robust method, whereas IPW and RA correspond to the inverse probability weighting and regression adjustment methods, respectively. 
The proposed doubly robust estimator of the CATEF works well in finite samples. 
As the theory indicates, the proposed estimator exhibits small bias provided that at least one of the regression model and the propensity score model is correctly specified. 
We find that the regression adjustment estimator is very precise when the regression model is correctly specified.
However, it suffers from substantial bias when the regression model is misspecified.
The inverse probability weighting estimator also suffers from model misspecification.
Moreover, its standard deviation is much larger than those of the doubly robust and regression adjustment estimators.
When both models are misspecified, all three estimators suffer from heavy bias.
The inverse probability weighting estimator has the largest RMSE because its distribution is more diverse than those of the other two estimators.
All the estimators have larger standard deviations when $x=1$ and $x=-1$ compared to those in other points. 
This is because the number of observations around $x=1$ or $x=-1$ is expected to be smaller than that around, for example, $x=0$ which is the center of the distribution.
On the other hand, the magnitude of the bias does not vary much across data points.
The standard error for the proposed doubly robust estimator is slightly smaller than the standard deviation, but the difference is not large.

Tables \ref{tab-cb-k10} and \ref{tab-cb-k30} summarize the finite sample properties of the proposed 
uniform confidence band.
 The results show that our uniform confidence band has a reasonably good coverage property provided that one of the models is correctly specified.
When both models are misspecified, the size distortion is heavy. 
We observe that the size distortion is heavier when the regression model is misspecified than that in the case of propensity score misspecification.
This result indicates that we should carefully model the regression part in order to obtain reliable confidence bands.
The average values of the 95\% critical values are around 2.75. Because the pointwise critical value is 1.96 and is much smaller than the uniformly valid critical value, it demonstrates the importance of the uniform property of confidence band.
The standard deviations of the critical values are small because they change only if the bandwidth changes. 
The confidence band based on the Gumbel approximation is very conservative.

The results of the Monte Carlo simulation confirm that the proposed doubly robust estimator indeed works well in finite samples provided that one of the regression and propensity score models is correctly specified.
The proposed uniform confidence band also has good coverage properties.

\section{An Empirical Application}\label{sec:EA}

We apply our uniformly valid confidence band for the CATEF for the effect of maternal smoking
 on birth weight where the argument of the CATEF is the mother's age.
Our aim here is to illustrate our confidence band in comparison with alternative confidence bands.
We first discuss the background of this application and the datasets used.
We use two different data sets: the dataset from Pennsylvania and that from North Carolina.
We then compute various confidence bands for the CATEF and discuss the results.

While the purpose of this application is to illustrate our uniformly valid confidence band and not to present new insights on the effect of smoking, it is still informative to discuss the background of this application. 
Many studies document that low birth weight is associated with prolonged negative effects on health and educational or labor market outcomes throughout life, although there has been a debate over its magnitude. See, e.g., \cite{AlmondCurrie11} for a review. 
Maternal smoking is considered to be the most important preventable negative cause of low birth weight \citep{Kramer87}.
There are many studies that evaluate the effect of maternal smoking on low birth weight \citep{AlmondCurrie11}.
The program evaluation approach is employed by, for example, 
\cite{AlmondChayLee05}, \cite{daVeigaWilder08} and
\cite{WalkerTekinWallace09}, and panel data analysis is carried out by 
\cite{Abrevaya06} and \cite{AbrevayaDahl08}.
Here, we are interested in how the effect of smoking changes across different age groups of mothers.
\cite{WalkerTekinWallace09} examine whether the effect of smoking is larger for teen mothers than for adult mothers and find mixed evidence.
\cite{Abrevaya-et-al-JBES} also consider this problem in their application.

\subsection{Pennsylvania data}
The first dataset consists of observations from white mothers in Pennsylvania in the USA.
The dataset is an excerpt from \cite{Cattaneo10} and is obtained from the STATA website (``\url{http://www.stata-press.com/data/r13/cattaneo2.dta}''). 
Note that the dataset was originally used in \cite{AlmondChayLee05}. 
We restrict our sample to white and non-Hispanic mothers, and the sample size is 3754.
The outcome of interest ($Y$) is infant birth weight measured in grams.
The treatment variable ($D$) is a binary variable that is equal to 1 if the mother smokes and 0 otherwise.
The set of covariates $\mathbf{Z}$ includes the mother's age, an indicator variable for alcohol consumption during pregnancy, an indicator for the first baby, the mother's educational attainment, an indicator for the first prenatal visit in the first trimester, the number of prenatal care visits, and an indicator for whether there was a previous birth where the newborn died.
We are interested in how the effect of smoking varies across different values of the mother's age.
Therefore, $\mathbf{X}$ is mother's age in this application.

To estimate the CATEF, we use linear regression models for the regression part and a logit model for propensity score.
The explanatory variables used in the regression models and the logit model 
consist of all the elements of $\mathbf{Z}$, the square of the mother's age, and the interaction terms between the mother's age and all other elements of $\mathbf{Z}$. 
We estimate the CATEF in the interval between ages 15 and 35.

We compute the following three 95\% confidence bands for the CATEF.
``Our CB'' is the confidence band proposed in this paper.
Because $\mathbf{X}$ is univariate in this application, we follow the algorithm in Section 3.
We use the Gaussian kernel.
The preliminary bandwidth ($\hat h$) is chosen by the method of \cite{Ruppert:Sheather:Wand:1995}.
``Gumbel CB'' is the confidence band in which $c(1-\alpha)$ in the algorithm is replaced by that based on the Gumbel approximation (see Remark \ref{rem-gumbel}).
``PW CB'' is a pointwise valid confidence band where we replace $c(1-\alpha)$ in the algorithm by the corresponding value from the standard normal distribution (i.e., $1.96$). 
This provides a valid confidence interval for each point of the CATEF.
However, its uniform coverage rate would be smaller than 95\%.

Figure 1 plots the estimated CATEF and the three 95\% confidence bands for the range between 
15 and 35 years of age. 
The figure also contains the average treatment effect estimate (AIPW estimate) for a reference.

The widths of the three confidence bands are substantially different.
The confidence band based on the Gumbel approximation provides the widest band and may not be very informative.
The confidence band that is valid only in a pointwise sense gives the narrowest band. 
This band is not uniformly valid and so may provide misleading information about the CATEF. On the other hand, this provides valuable information if we are interested at a particular point of the CATEF.
The confidence band we propose lay between ``Gumbel CB'' and ``PW CB''. 
While this band is wider than ``PW CB'', it is much narrower than ``Gumbel CB'' and is uniformly asymptotically valid. 
We see from this figure that our confidence band is informative while being uniformly valid.

The estimated CATEF is decreasing from 15 to around 25 years of age. It is rather stable for the range above 25 years of age. 
All confidence bands indicate that the CATEF is estimated imprecisely near the ends of the range.
Nonetheless, the estimated CATEF indicates that smoking may not have a strong impact when the mother is young.
The CATEF is estimated relatively precisely in the middle of the range.
For the range between 20 and 30 years of age, even the band based on the Gumbel approximation, which is the widest, does not contain 0.
This result provides robust evidence that smoking has a negative impact on birth weight at least for mothers who are 20 to 30 years old.
In this particular dataset, the statistical evidence against a constant smoking effect is somewhat weak. 
The confidence band that is valid only in a pointwise sense may provide an impression that the smoking effect depends on the mother's age.
However, the uniformly valid confidence band that we propose marginally contains the straight line 
that is equal to the ATE estimate.
This result illustrates that there is a caveat when we use pointwise confidence intervals, as well as the importance of using uniformly valid confidence bands.

\subsection{North Carolina data}

The second dataset is based on the records between 1988 and 2002 by the North Carolina State Center Health Services. 
This dataset is used in \cite{Abrevaya-et-al-JBES} and obtained from Robert Lieli's website (``\url{http://www.personal.ceu.hu/staff/Robert_Lieli/cate-birthdata.zip}).
We restrict our sample to white and first--time mothers, and the sample size is 433,558.
As in the case of the Pennsylvania data, the outcome is infant birth weight measured in grams
and the treatment variable is an indicator for smoking status.
The set of covariates $\mathbf{Z}$ includes those used in the analysis of the Pennsylvania data, except an indicator for the first baby because we focus on first--time mothers, and in addition, it includes indicators for gestational diabetes, hypertension, amniocentesis and ultra sound exams.
Again, $\mathbf{X}$ is mother's age in this application.
The specification for the estimation of the CATEF is the same as before.

The purpose of using this much larger dataset is to examine the effect of the sample size.
Our method involves nonparametric kernel regression and it might require a large sample size to yield a reliable result. 
For example, the result from the Pennsylvania data indicates that the effect of smoking is very small for very young mothers. 
One might argue that such a result is an artifact of small sample size.
The other issue is that the confidence bands obtained using the Pennsylvania data are somewhat wide.
We hope that using this larger dataset provides us with narrower confidence bands and more informative statistical results.

Figure 2 plots the estimated CATEF and the three 95\% confidence bands for the range between 
15 and 35 years of age. 
Note that the scale of the vertical axis is different from Figure 1.
We now obtain much narrower confidence bands. 
The widths of the three (uniform, point-wise and Gumbel) confidence bands 
are still different.
The estimated CATEF for young mothers is negative and statistically different from 0. 
The previous result that it is close to 0 may be considered as an artifact of small sample size.
The estimated CATEF is decreasing from around 17 to around 29 years of age. 
For the range above 30 years of age, we obtain relatively wide confidence bands. 
We reject the null hypothesis of no effect of smoking on birth weights uniformly over 15-35 years of age.
These confidence bands do not support the hypothesis that the CATEF is constant because the ATE line exceeds the confidence bands. 

One might argue that the difference in the results may stem from the fact that the North Carolina data contains richer information and we use a larger set of covariates.
We reexamine the North Carolina data based on the same set of covariates as that for the Pennsylvania data, except an indicator for the first baby.
Figure 3 plots the estimated CATEF and confidence bands obtained using this set of covariates.
The results in Figure 3 are qualitatively very similar to those in Figure 2.
We thus believe that the difference between the results from the Pennsylvania data and the North Carolina data are not from the difference in the covariates but from the difference in the sizes of these two samples.

We thus interpret our findings to indicate that the different results come from the difference in sample size yet our confidence bands reasonably quantify the uncertainty from small sample size.
While two data-sets yield different estimates of CATEF, the confidence bands from the Pennsylvania data include the estimated CATEF and the confidence bands from the North Carolina data.

While we use the same data set as that used in \cite{Abrevaya-et-al-JBES}, it is somewhat difficult to compare their results with ours because of differences in the implementations. 
In particular, the bandwidths are very different. 
Our choice of bandwidth is around 0.2, while theirs are between 1.4--11.2. 
Nonetheless, we make several remarks. 
Using small bandwidths (1.4 and 2.8), \cite{Abrevaya-et-al-JBES} observe almost no effect for young mothers and a large negative effect for 25--30 years old mothers.
We do not observe such a large difference in the effect across different age groups.
Our confidence band is as tight as their confidence band obtained with bandwidth 11.2 even though we use a much smaller bandwidth and our confidence band is uniform. 
This is possibly because we use an AIPW method which yields a more efficient estimate than an IPW method does.

\section{Conclusion}\label{sec:C}
 In this paper, we propose a doubly robust method for estimating the CATEF. 
We consider the situation where a high-dimensional vector of covariates is needed for identifying the average treatment effect but the covariates of interest are of much lower dimension. 
Our proposed estimator is doubly robust and does not suffer from the curse of dimensionality.
We propose a uniform confidence band that is easy to compute, and we illustrate its usefulness via Monte Carlo experiments and an application to the effects of smoking on birth weights. 

There are a few topics to be explored in the future. 
First, it would be useful to consider the issue of asymptotic biases of the proposed estimator without relying on undersmoothing. For example, it might be possible to extend the approach 
of \cite{hall2013} that avoids undersmoothing for our purposes.
Second, it would be an interesting exercise to develop a method for estimating the quantile treatment effects conditional on covariates. 
Third, it is possible to extend our approach to the local average treatment effect.
As mentioned in the Introduction, 
\cite{Ogburn-et-al-JRSSB} consider conditioning on $\mathbf{Z}$ to achieve identification, but they estimate the local average treatment effect, say LATE($\mathbf{X}$), as a function of $\mathbf{X}$. However, their specification of LATE($\mathbf{X}$) is parametric.
Our approach can be adapted to specify LATE($\mathbf{X}$) nonparametrically and to develop a corresponding uniform confidence band. 
Fourth, this paper does not cover marginal treatment effects that can be identified using the method of local instrumental variables developed by \citet{Heckman:Vytlacil:99,Heckman:Vytlacil:05}. It would be interesting to develop a uniform confidence band for the marginal treatment effects.

\appendix

\section{The direct plug-in bandwidth selector of \cite{Ruppert:Sheather:Wand:1995}} \label{ap-bandwidth}
In this section, we give a brief description of the direct plug-in bandwidth selector of \cite{Ruppert:Sheather:Wand:1995} for local linear regression.
We focus on the case of the Gaussian kernel and univariate regressor. 
Note that this bandwidth can be computed with the ``\textrm{dpill}'' function in the ``KernSmooth'' package for R \citep{kernsmooth}.

In the following, we denote the dependent variable by $\psi_i$ and the regressor by $X_i$. 
We consider estimating $\mathbb{E}[ \psi | X= x]$ for $x \in [a, b]$ for some $a$ and $b$. 
In our implementation, we use $a = \min_{1 \le i \le n} X_i$ and $b = \max_{1\le i \le n } X_i$.

\paragraph{Step 1:}
We divide the sample into $N$ blocks and estimate a quartic regression model for each block.
The number of blocks is chosen by minimizing the Mallows' $C_p$:
\begin{align*}
 C_p (N) = \frac{RSS(N)}{RSS(N_{\max})}(n- N_{\max}) - (n-10N),
\end{align*}
where $RSS (N)$ is the residual sum of squares based on a blocked quartic fit over $N$ blocks, and 
\begin{align*}
 N_{\max} = \max \{ \min ( \lfloor n/20 \rfloor , 5), 1\}.
\end{align*}
Let $\hat m_Q^{(2)}$ and $\hat m_Q^{(4)}$ be the estimates of the second and fourth derivative of the regression function from the blocked quartic fit.
Let 
\begin{align*}
 \hat \theta_{24}^Q = \frac{1}{n} \sum_{i=1}^n \sum_{j=1}^N \hat m_Q^{(2)} (X_i)  \hat m_Q^{(4)} (X_i) 
\mathbf{1}_{\{ X_i \in \mathcal{X}_j\}},
\end{align*}
where $\mathcal{X}_j$ is the set of $X_i$ belonging to the $j$-th block.
Let $\hat m_Q  $ be the estimated regression curve from the blocked quartic fit.
Let
\begin{align*}
 \hat \sigma_{Q}^2 
= \frac{1}{n-5N} \sum_{i=1}^n \sum_{j=1}^N (Y_i - \hat m_Q (X_i))^2 \mathbf{1}_{\{ X_i \in \mathcal{X}_j\}}.
\end{align*}

\paragraph{Step 2:}
We estimate a local cubic regression model using the following bandwidth:
\begin{align*}
 \hat g_1
= C_2 (K) \left[ \frac{\hat \sigma_Q^2 (b-a)}{|\hat \theta_{24}^Q|} n\right]^{1/7},
\end{align*}
where
\begin{align*}
 C_2 (K)
=
\begin{cases}
 \{ 3/ (8 \sqrt{\pi})\}^{1/7} & \text{if } \hat \theta_{24}^Q <0, \\
 \{ 15/ (16 \sqrt{\pi})\}^{1/7} & \text{if } \hat \theta_{24}^Q >0.
\end{cases}
\end{align*}
Let $\hat m_C^{(2)}$ be the estimate of second derivative of the regression function from the local cubic regression.
Let 
\begin{align*}
 \hat \theta_{22}
= \frac{1}{n} \sum_{i=1}^n ( \hat m_C^{(2)} (x_i) )^2
\mathbf{1}_{\{ 0.95 a + 0.05 b < X_i < 0.05 a + 0.95b   \}}.
\end{align*}

We estimate a local linear regression model using the following bandwidth:
\begin{align*}
 \hat g_2 = \left\{ 4 \left( \frac{1}{2} + 2 \sqrt{2} - \frac{4}{3} \sqrt{3} / \sqrt{2\pi}\right)\right\}^{1/9}
\left[ \frac{\hat \sigma_{Q}^4 (b-a )}{\hat \theta_{22}^2 n^2}\right]^{1/9}.
\end{align*}
Let $\hat m_L$ be the estimated regression curve from this local linear regression.
Let 
\begin{align*}
 \hat \sigma^2 
= \frac{1}{n - 2 \sum_{i=1}^n w_{ii} + \sum_{i=1}^n \sum_{j=1}^n w_{ij}^2} \sum_{i=1}^n (\psi_i- \hat m (X_i))^2,
\end{align*}
where $w_{ij}$ is the $(1, j)$-th element of $(X_{1,i}^{\top} W_i X_{1,i})^{-1} X_{1,i}^{\top} W_i$,
$X_{1,i}$ is the $n \times 2$ matrix whose first column is a vector of ones and the $j$-th element of whose second column is $X_j - X_i$, $W_i$ is the diagonal matrix whose $j$-th element is $K\{ (X_j- X_i )/\hat g_2\} / \hat g_2$ and $K$ is the kernel function.

\paragraph{Step 3:}

The bandwidth is computed as:
\begin{align*}
 \hat h =  \left( \frac{1}{2 \sqrt{\pi}} \right)^{1/5} 
\left[ \frac{ \hat \sigma^2 (b-a)}{ \hat \theta_{22} n }\right]^{1/5}.
\end{align*}

\section{Proofs}\label{sec:A}

\begin{proof}[Proof of Lemma \ref{iden-lem}]
Because $DY = DY_1$ and $Y_1$ and $D$ are independent of each other conditional on $\mathbf{Z}$, write:
\begin{align}\label{double-robust-lem-eq}
\mathbb{E}\left[ \psi _{1}(\mathbf{W},\alpha _{10},\beta_0 )|\mathbf{X}=\mathbf{x} \right]
&= \mathbb{E}\left[  \frac{\mathbb{E}\left[ D|\mathbf{Z}\right]  \mathbb{E}\left[Y_1|\mathbf{Z}\right] }{\pi (\mathbf{Z},\beta_0 )}-\frac{\mathbb{E}\left[D|\mathbf{Z}\right] -\pi
(\mathbf{Z},\beta_0 )}{\pi (\mathbf{Z},\beta_0 )}\mu _{1}(\mathbf{Z},\alpha _{10})  \bigg|\mathbf{X}=\mathbf{x} \right].
\end{align}
Suppose that $\beta_0$ satisfies $\mathbb{E}\left[D|\mathbf{Z}\right] = \pi
(\mathbf{Z},\beta_0 )$ almost surely. Then the right-hand side of \eqref{double-robust-lem-eq} reduces to: 
\begin{align*}
\mathbb{E}\left[    \mathbb{E}\left[Y_1|\mathbf{Z}\right] |\mathbf{X}=\mathbf{x} \right] = \mathbb{E}\left[    Y_1 |\mathbf{X}=\mathbf{x} \right].
\end{align*}
Suppose now that $\alpha_{10}$ satisfies
$\mathbb{E}\left[Y_1|\mathbf{Z}\right] = \mu _{1}(\mathbf{Z},\alpha _{10})$ almost surely. 
Then the right-hand side of \eqref{double-robust-lem-eq} again reduces to: 
\begin{align*}
\mathbb{E}\left[    \mu _{1}(\mathbf{Z},\alpha _{10}) |\mathbf{X}=\mathbf{x} \right] = \mathbb{E}\left[    Y_1 |\mathbf{X}=\mathbf{x} \right].
\end{align*}
Analogously, we have $\mathbb{E}\left[ \psi _{0}(\mathbf{W},\alpha _{00},\beta_0 )|\mathbf{X}=\mathbf{x} \right] = \mathbb{E}\left[    Y_0 |\mathbf{X}=\mathbf{x} \right]$.
\end{proof}

The remainder of the appendix gives the proof of Theorem \ref{main-uniform}.
We first establish the linear expansion of the local linear estimator.

\begin{lem}\label{KLX-lemma}
\begin{align*}
  \sup_{\mathbf{x} \in \mathcal{I}}  \sqrt{n h_n^d} \left| \frac{\hat{g}(\mathbf{x}) -  g(\mathbf{x})}{\hat{s}(\mathbf{x})} - \frac{1}{n h_n^d s_n(\mathbf{x})}   \sum_{i=1}^n \frac{U_i}{f_\mathbf{X}(\mathbf{x})} \mathbf{K} \left( \frac{\mathbf{X}_i - \mathbf{x}}{h_n} \right)\right| 
  = O_p \left(  n^{-c} \right)
\end{align*}
for some positive constant $c > 0$.
\end{lem}

\begin{proof}[Proof of Lemma \ref{KLX-lemma}]
Let $\hat{\bm{\Psi}}$ and $\bm{\Psi_0}$ denote the $n$-dimensional vectors such that
$\hat{\bm{\Psi}} = \{\psi (\mathbf{W}_{i},\hat \theta) \}_{i=1}^n$ and 
$\bm{\Psi_0} = \{\psi (\mathbf{W}_{i},\theta_0) \}_{i=1}^n$, respectively.
Let $\bm{\Gamma}(\bm{x})$ be the $n \times (d+1)$ matrix whose $i$-th row is $[1, (\mathbf{X}_i - \mathbf{x})^\top]$,
$\bm{\Omega}(\bm{x})$  the $n$-dimensional diagonal matrix whose $i$-th element is 
$h_n^{-1} \mathbf{K} \left[ (\mathbf{X}_i - \mathbf{x})/{h_n} \right]$,
$\bm{G} := [g(\bm{X}_i)]_{i=1}^n$ the  $n$-dimensional vector of regression functions evaluated at data points,
and 
$\bm{U} := (U_i)_{i=1}^n$ the $n$-dimensional vector of regression errors.
Let $\mathbf{e}_1$ denote the $(d+1) \times 1$ vector whose  first element is one and all others are zeros.
Write
\begin{align*}
\hat{g}(\bm{x}) - {g}(\bm{x}) &= T_{n1}(\bm{x}) + T_{n2}(\bm{x}) + R_{n1} (\bm{x}), 
\end{align*}
where 
\begin{align*}
T_{n1}(\bm{x}) &= \mathbf{e}_1^\top \left [\bm{\Gamma}(\bm{x})^\top \bm{\Omega}(\bm{x}) \bm{\Gamma}(\bm{x}) \right]^{-1} 
\bm{\Gamma}(\bm{x})^\top \bm{\Omega}(\bm{x}) \bm{U}, \\
T_{n2}(\bm{x}) &= \mathbf{e}_1^\top \left [\bm{\Gamma}(\bm{x})^\top \bm{\Omega}(\bm{x}) \bm{\Gamma}(\bm{x}) \right]^{-1} 
\bm{\Gamma}(\bm{x})^\top \bm{\Omega}(\bm{x}) \bm{G}, \\
R_{n} (\bm{x}) &= 
\mathbf{e}_1^\top \left [\bm{\Gamma}(\bm{x})^\top \bm{\Omega}(\bm{x}) \bm{\Gamma}(\bm{x}) \right]^{-1} 
\bm{\Gamma}(\bm{x})^\top \bm{\Omega}(\bm{x}) \left( \hat{\bm{\Psi}} - \bm{\Psi_0} \right).
\end{align*}%

We first consider the leading stochastic term $T_{n1}(\bm{x})$.
As in equation (2.10) of \cite{Ruppert/Wand:94}, we have that 
\begin{align}\label{before-inverse-expression}
\begin{split}
&n^{-1} \bm{\Gamma}(\bm{x})^\top \bm{\Omega}(\bm{x}) \bm{\Gamma}(\bm{x}) \\
&= \left(
\begin{array}{cc}
\frac{1}{n h_n^d} \sum_{i=1}^n  \mathbf{K} \left( \frac{\mathbf{X}_i - \mathbf{x}}{h_n} \right) & 
\frac{1}{n h_n^d} \sum_{i=1}^n  \mathbf{K} \left( \frac{\mathbf{X}_i - \mathbf{x}}{h_n} \right) (\mathbf{X}_i - \mathbf{x})^\top \\
\frac{1}{n h_n^d} \sum_{i=1}^n  \mathbf{K} \left( \frac{\mathbf{X}_i - \mathbf{x}}{h_n} \right)  
(\mathbf{X}_i - \mathbf{x})&\
\frac{1}{n h_n^d} \sum_{i=1}^n  \mathbf{K} \left( \frac{\mathbf{X}_i - \mathbf{x}}{h_n} \right) 
(\mathbf{X}_i - \mathbf{x}) (\mathbf{X}_i - \mathbf{x})^\top
\end{array}
\right).
\end{split}
\end{align}
By the standard empirical process results 
(see e.g., \cite{Pollard:84} or \cite{VanDerVaart/Wellner:96})
combined with the usual change of variables, we have that uniformly in $\mathbf{x} \in \mathcal{I}$,
 \begin{align*}
\frac{1}{n h_n^d} \sum_{i=1}^n  \mathbf{K} \left( \frac{\mathbf{X}_i - \mathbf{x}}{h_n} \right) 
& = f_{\mathbf{X}}(\mathbf{x}) + o(h_n) +  O_p \left[  \left(\frac{\log n}{n h_n^d}\right)^{1/2} \right], 
\\ 
\frac{1}{n h_n^d} \sum_{i=1}^n  \mathbf{K} \left( \frac{\mathbf{X}_i - \mathbf{x}}{h_n} \right)  
(\mathbf{X}_i - \mathbf{x}) 
& =  h_n^2 \frac{\partial f_{\mathbf{X}} (\mathbf{x})}{\partial \mathbf{x} } \mu_2(K) 
+ o(h_n^2) + O_p \left[  h_n \left(\frac{\log n}{n h_n^d}\right)^{1/2} \right], \\
\frac{1}{n h_n^d} \sum_{i=1}^n  \mathbf{K} \left( \frac{\mathbf{X}_i - \mathbf{x}}{h_n} \right) 
(\mathbf{X}_i - \mathbf{x}) (\mathbf{X}_i - \mathbf{x})^\top
& =  h_n^2  f_{\mathbf{X}} (\mathbf{x})\mu_2(K) 
+ o(h_n^2) + O_p \left[  h_n^2 \left(\frac{\log n}{n h_n^d}\right)^{1/2} \right], 
\end{align*}
where $\mu_2(K) := \int u^2 K(u) du$.

Throughout the remainder of the proof, we let $r_n(\mathbf{x}) = O_p( n^{-c})$ , uniformly in $\bm{x}$, be a sequence that can be 
different in different places for some constant $c > 0$.  
Then as in (2.11) of \cite{Ruppert/Wand:94}, we have that
\begin{align}\label{inverse-expression}
\begin{split}
&\left [ n^{-1} \bm{\Gamma}(\bm{x})^\top \bm{\Omega}(\bm{x}) \bm{\Gamma}(\bm{x}) \right]^{-1} \\
&= \left(
\begin{array}{cc}
f_{\mathbf{X}}(\mathbf{x})^{-1}[1 + r_n(\mathbf{x})]  &  -  [{\partial f_{\mathbf{X}} (\mathbf{x})}/{\partial \mathbf{x} }]^\top  f_{\mathbf{X}}(\mathbf{x})^{-2}[1 + r_n(\mathbf{x})]  \\
-  [{\partial f_{\mathbf{X}} (\mathbf{x})}/{\partial \mathbf{x} }]  f_{\mathbf{X}}(\mathbf{x})^{-2}[1 + r_n(\mathbf{x})]  &     \left[ \mu_2(K)  f_{\mathbf{X}}(\mathbf{x}) h_n^2 \mathbf{I}_{d} \right]^{-1}[1 + r_n(\mathbf{x})]
\end{array}
\right),
\end{split}
\end{align}
where $\mathbf{I}_{d}$ is the $d$-dimensional identity matrix.
The little $o_p(\cdot)$ terms in equation (2.11) of \cite{Ruppert/Wand:94} are pointwise;
however, \eqref{inverse-expression} holds  uniformly in $\mathbf{x} \in \mathcal{I}$
with polynomially decaying terms $r_n(\mathbf{x})$  under our assumptions.

Let $\bm{\Gamma}_i(\bm{x}) := [1, (\mathbf{X}_i - \mathbf{x})^\top ]^\top$.
Since
\begin{align*}
n^{-1} \bm{\Gamma}(\bm{x})^\top \bm{\Omega}(\bm{x}) \bm{U} 
&= \frac{1}{n h_n^d} \sum_{i=1}^n U_i 
 \mathbf{K} \left( \frac{\mathbf{X}_i - \mathbf{x}}{h_n} \right) \bm{\Gamma}_i(\bm{x}),
 \end{align*}
we have by \eqref{inverse-expression}  that $T_{n1}(\bm{x}) = T_{n11}(\bm{x})  + T_{n12}(\bm{x})$, where
 \begin{align*}
T_{n11}(\bm{x}) &= \frac{1}{n h_n^d f_{\mathbf{X}}(\mathbf{x})} \sum_{i=1}^n U_i 
 \mathbf{K} \left( \frac{\mathbf{X}_i - \mathbf{x}}{h_n} \right) [1 + r_n(\mathbf{x})], \\
T_{n12}(\bm{x})  &=
-\frac{1}{n h_n^d [f_{\mathbf{X}}(\mathbf{x})]^2} \sum_{i=1}^n U_i 
 \mathbf{K} \left( \frac{\mathbf{X}_i - \mathbf{x}}{h_n} \right) 
 \left[ \frac{\partial f_{\mathbf{X}} (\mathbf{x})}{\partial \mathbf{x} } \right]^\top (\mathbf{X}_i - \mathbf{x}) [1 + r_n(\mathbf{x})].
\end{align*}
Again using the standard empirical process result and the method of change of variables, 
\begin{align*}
T_{n11}(\bm{x}) = O_p \left[  \left( \frac{\log n}{n h_n^d} \right)^{1/2}  \right]
\; \text{ and } \;
T_{n12}(\bm{x}) = O_p \left[  h_n \left( \frac{\log n}{n h_n^d} \right)^{1/2}  \right]
\end{align*}
uniformly in $\mathbf{x} \in \mathcal{I}$.
Therefore, we have shown that
\begin{align}\label{bahadur-lemma}
T_{n1}(\bm{x}) &= \frac{1}{n h_n^d f_{\mathbf{X}}(\mathbf{x})} \sum_{i=1}^n U_i 
 \mathbf{K} \left( \frac{\mathbf{X}_i - \mathbf{x}}{h_n} \right)[1 + r_n(\mathbf{x})].
\end{align}

We now move on the other remainder terms.
The proof of Theorem 2.1 (in particular, equation (2.3)) of \cite{Ruppert/Wand:94} implies that 
$T_{n2}(\bm{x}) = O ( h_n^{2} )$ uniformly in $\mathbf{x} \in \mathcal{I}$.
The condition that $n h_n^{d + 4} \rightarrow 0$ at a polynomial rate in $n$ corresponds to the undersmoothing condition.
It is straightforward to show that
$(n h_n^d)^{1/2} R_n(\bm{x}) = O ( n^{-c} )$ uniformly in $\mathbf{x}$ for some constant $c > 0$ due to Assumption \ref{main-assumption}(9) that 
$$\max\left\{  (n h_n^d)^{1/2}|\psi (\mathbf{W}_{i},\hat \theta) - \psi (\mathbf{W}_{i},\theta_0) | i=1,\dots,n \right\} = O_p( n^{-c})$$ for some constant $c > 0$ 

Note that by conditions (7) and (8) of Assumption \ref{main-assumption}, we have that 
$\inf_{n \geq 1} \inf_{\mathbf{x} \in \mathcal{I}} s_n(\mathbf{x}) > 0$ and 
$\sup_{\mathbf{x} \in \mathcal{I}} \left| \hat s^2 (\mathbf{x}) - s_n^2(\mathbf{x}) \right| = O_p( n^{-c})$. Hence, the lemma follows from \eqref{bahadur-lemma} immediately.
\end{proof}

Define: 
\begin{align*}
T_n(\mathbf{x}) &\equiv \frac{1}{n h_n^d} \sum_{i=1}^n U_i \mathbf{K} \left( \frac{\mathbf{X}_i - \mathbf{x}}{h_n} \right) \ \ \text{and} \ \
c_n(\mathbf{x}) 
\equiv \left\{ \frac{1}{h_n^{d}} \mathbb{E}\left[ U^2 \mathbf{K}^2 \left( \frac{\mathbf{X} - \mathbf{x}}{h_n} \right) \right] \right\}^{-1/2}.
\end{align*}
Note that $c_n(\mathbf{x}) = [f_\mathbf{X}(\mathbf{x}) s_n(\mathbf{x})]^{-1}$.
By  Lemma \ref{KLX-lemma}: 
\begin{align*}
  \max_{\mathbf{x} \in \mathcal{I}}   \sqrt{n h_n^d} \left| \frac{\hat{g}(\mathbf{x}) -  g(\mathbf{x})}{\hat{s}(\mathbf{x})} - c_n(\mathbf{x})  T_n(\mathbf{x}) \right| = O_p \left(  n^{-c} \right).
\end{align*}

We now use the result of \cite{CCK:2014} to obtain Gaussian approximations.
Define:
\begin{align}\label{sup-stat}
W_n  \equiv \sup_{\mathbf{x} \in \mathcal{I}} c_n(\mathbf{x}) \sqrt{n h_n^d} \left[ T_n(\mathbf{x}) - \mathbb{E} T_n(\mathbf{x})\right].
\end{align}
\cite{CCK:2014} established an approximation of
$W_n$ by a sequence of suprema of Gaussian processes. For each $n \geq 1$, let 
$\tilde{B}_{n,1}$ be a centered Gaussian process indexed by $\mathcal{I}$ with covariance function:
\begin{align}\label{cov-kl}
\begin{split}
& \mathbb{E}[\tilde{B}_{n,1}(\mathbf{x}) \tilde{B}_{n,1}(\mathbf{x}')]= \\
& h_n^{-d} c_n(\mathbf{x}) c_n(\mathbf{x}') 
\text{Cov} \left[ U^2 \mathbf{K} \left( \frac{\mathbf{X} - \mathbf{x}}{h_n} \right) \mathbf{K} \left( \frac{\mathbf{X} - \mathbf{x}'}{h_n} \right) \right].
\end{split}
\end{align}
Proposition 3.2 of \cite{CCK:2014} establishes the following approximation result.

\begin{lem}
\label{cck-thm}
 Let Assumption \ref{main-assumption} hold. Then for every $n \geq 1$, there is a tight Gaussian random variable $\tilde{B}_{n,1}$ in $\ell^\infty(\mathcal{I})$ with mean zero and covariance function \eqref{cov-kl},
and there is a sequence $\tilde{W}_{n,1}$ of random variables such that 
$\tilde{W}_{n,1} =_d \sup_{\mathbf{x} \in \mathcal{I}} \tilde{B}_{n,1}(\mathbf{x})$ and as $n \rightarrow \infty$: 
\begin{align*}
|W_n - \tilde{W}_{n,1}| = O_\mathbb{P} 
\left\{
(n h_n^d)^{-1/6} \log n + (n h_n^d)^{-1/4} \log^{5/4} n 
+ (n^{1/2} h_n^d)^{-1/2} \log^{3/2} n 
\right\}.
\end{align*}
\end{lem}

\begin{proof}[Proof of Lemma \ref{cck-thm}]
To apply Proposition 3.2 of \cite{CCK:2014}, we first note that Assumption \ref{main-assumption} implies that all the regularity conditions for Proposition 3.2 of \cite{CCK:2014} are satisfied. They are:
\begin{enumerate}
\item $\sup_{\mathbf{x} \in \mathbb{R}^d} \mathbb{E}[U^4|\mathbf{X}=\mathbf{x}] < \infty$.
\item $\mathbf{K}(\cdot)$ is a bounded and continuous kernel function on $\mathbb{R}^d$,
and such that the class of functions 
$\mathcal{\mathbf{K}}  \equiv \{ \mathbf{t} \mapsto \mathbf{K}(h\mathbf{t} + \mathbf{x}): h > 0, \mathbf{x} \in \mathbb{R}^d \}$ is a VC type
with envelope $\norm{ \mathbf{K} }_\infty$.
\item The distribution of $\mathbf{X}$ has a bounded Lebesgue density $p(\cdot)$ on $\mathbb{R}^d$.
\item $h_n \rightarrow 0$ and $\log (1/h_n) = O(\log n)$ as $n \rightarrow \infty$.
\item $C_{\mathcal{I}} \equiv \sup_{n \geq 1} \sup_{\mathbf{x} \in \mathcal{I}} |c_n(\mathbf{x})| < \infty$.
Moreover, for every fixed $n \geq 1$ and for every $\mathbf{x}_m \in \mathcal{I} \rightarrow \mathbf{x}
\in \mathcal{I}$ pointwise, $c_n(\mathbf{x}_m) \rightarrow c_n(\mathbf{x})$.
\end{enumerate}
Then the desired result is an immediate consequence of Proposition 3.2 of 
\cite{CCK:2014} with a singleton set $\mathcal{G} = \{ U \}$ and with $q = 4$ (using their notation) in verifying condition
(B1)' of  \cite{CCK:2014}.
\end{proof}

We now show that the Gaussian field obtained in Lemma \ref{cck-thm} can be further approximated by a stationary Gaussian field.

\begin{lem}\label{step2-thm}
Let Assumption \ref{main-assumption} hold. Then for every $n \geq 1$, there is a tight Gaussian random variable $\tilde{B}_{n,2}$ in $\ell^\infty(\mathcal{I}_n)$ with mean zero and covariance function: 
\begin{align*}
 \mathbb{E}[\tilde{B}_{n,2}(\mathbf{s}) \tilde{B}_{n,2}(\mathbf{s}')]= \rho_d(\mathbf{s}-\mathbf{s}')
\end{align*}
for $\mathbf{s},\mathbf{s}' \in \mathcal{I}_n  \equiv h_n^{-1} \mathcal{I}$, and there is a sequence of 
random variables such that 
$\tilde{W}_{n,2} =_d \sup_{\mathbf{x} \in \mathcal{I}} \tilde{B}_{n,2}(h_n^{-1} \mathbf{x})$ and as $n \rightarrow \infty$: 
\begin{align*}
|\tilde{W}_{n,1} - \tilde{W}_{n,2}| = O_\mathbb{P} 
\left(
h_n \sqrt{\log h_n^{-d}} 
\right).
\end{align*}
\end{lem}

\begin{proof}[Proof of Lemma \ref{step2-thm}]
This lemma can be proved as in the proof of Lemma 3.4 of \cite{Ghosal/Sen/vanderVaart:00}. 
Let: 
\begin{align*}
\phi_{n, \mathbf{x}} (U_i,\mathbf{X}_i)  
&:=  \left\{\mathbb{E}\left[ U^2 \mathbf{K}^2 \left( \frac{\mathbf{X} - \mathbf{x}}{h_n} \right) \right] \right\}^{-1/2} U_i \mathbf{K} \left( \frac{\mathbf{X}_i - \mathbf{x}}{h_n} \right), \\
\varphi_{n, \mathbf{x}} (U_i,\mathbf{X}_i)  
&:=  \left\{h_n^{d} \mathbb{E}\left[ U^2 |  \mathbf{X}_i \right] f_{\mathbf{X}}( \mathbf{X}_i) \int   \mathbf{K}^2( \mathbf{u} )  d\mathbf{u}  \right\}^{-1/2} U_i \mathbf{K} \left( \frac{\mathbf{X}_i - \mathbf{x}}{h_n} \right).
\end{align*}
As in Remark 8.3 of \cite{Ghosal/Sen/vanderVaart:00} and in the proof of Lemma 3.4 of \cite{Ghosal/Sen/vanderVaart:00}, we can regard Gaussian processes $\tilde{B}_{n,1}$ and $\tilde{B}_{n,2}$ as Brownian bridges $B_n(\phi_{n, \mathbf{x}})$ and $B_n(\varphi_{n, \mathbf{x}})$, respectively, in the sense that $\mathbb{E} B_n(g) = 0$ and 
$\mathbb{E} [B_n(g) B_n(g')] = \text{cov}(g,g')$ for $g =  \phi_{n, \mathbf{x}}, g' =  \phi_{n, \mathbf{x'}}$ or  $g =  \varphi_{n, \mathbf{x}}, g =  \varphi_{n, \mathbf{x'}}$.

Define $\delta_{n}(\mathbf{x}) := B_n(\phi_{n, \mathbf{x}}) - B_n(\varphi_{n, \mathbf{x}})$.
Note that $\delta_{n}(\mathbf{x})$ is also a mean zero Gaussian process with: 
\begin{align*}
\mathbb{E}\left[ \delta_{n}(\mathbf{x}) \delta_{n}(\mathbf{x'}) \right] = \mathbb{E}\left[ \{ \phi_{n, \mathbf{x}}(U,\mathbf{X}) - \varphi_{n, \mathbf{x}} (U,\mathbf{X}) \} \{ \phi_{n, \mathbf{x'}}(U,\mathbf{X}) - \varphi_{n, \mathbf{x'}} (U,\mathbf{X}) \} \right].
\end{align*}

Note that:
\begin{align*}
\mathbb{E}\left[ \{ \delta_{n}(\mathbf{x}) \}^2 \right] &= \int \bigg( \left\{ \int  \mathbb{E}\left[ U^2 |  \mathbf{X} = \mathbf{x} + h \mathbf{u} \right] \mathbf{K}^2( \mathbf{u} ) f_{\mathbf{X}}( \mathbf{x} + h \mathbf{u})  d\mathbf{u} \right\}^{-1/2} \\
&- \left\{  \mathbb{E}\left[ U^2 |  \mathbf{X} = \mathbf{x} + h \mathbf{t} \right]  f_{\mathbf{X}}( \mathbf{x} + h \mathbf{t}) \int \mathbf{K}^2\left( \mathbf{u} \right) d\mathbf{u} \right\}^{-1/2} \bigg)^2 \\
&\times   \mathbb{E}\left[ U^2 |  \mathbf{X} = \mathbf{x} + h \mathbf{t} \right] \mathbf{K}^2 \left( \mathbf{t} \right)  f_{\mathbf{X}}( \mathbf{x} + h \mathbf{u})  d\mathbf{t} \\
&= O(h_n^2),
\end{align*}
because $\mathbf{x} \mapsto \mathbb{E}\left[ U^2 | \mathbf{X} =\mathbf{x} \right]f_{\mathbf{X}}(\mathbf{x})$ is Lipschitz continuous.
Thus, the $L_2$-diameter of $\delta_{n}(\cdot)$ is $O(h_n)$. 
In addition, we can show that there exists a constant $C > 0$ such that: 
\begin{align*}
\mathbb{E}\left[ \{ \delta_{n}(\mathbf{x}) - \delta_{n}(\mathbf{x'}) \}^2 \right] 
\leq C h^{-2} \norm{\mathbf{x} - \mathbf{x'}}^2.
\end{align*}
Then arguments similar to those used in the proof of Lemma 3.4 of 
\cite{Ghosal/Sen/vanderVaart:00} yield the desired result.

\end{proof}

\begin{proof}[Proof of Theorem \ref{main-uniform}]
First note that $a_n = O( \sqrt{\log n}\,)$ because $h_n =C n^{-\eta}$. Lemmas \ref{cck-thm} and \ref{step2-thm} together imply that: 
\begin{align*}
  \max_{\mathbf{x} \in \mathcal{I}}   \left| \frac{\hat{g}(\mathbf{x}) -  g(\mathbf{x})}{\hat{s}(\mathbf{x})} - \tilde{B}_{n,2}(h_n^{-1} \mathbf{x}) \right| = o_p \left( a_n \right).
\end{align*}
Note that $\tilde{B}_{n,2}$, defined in Theorem \ref{step2-thm}, is a homogeneous Gaussian field with zero mean and the covariance function $\rho_d(\mathbf{s})$.
Because of the assumption on $K(\cdot)$, the covariance function $\rho_d(\mathbf{s})$ has finite support and is six times differentiable.
The latter property implies that the Gaussian process $\tilde{B}_{n,2}$ is 
three times differentiable in the mean square sense (see, e.g., Chapter 4 of \cite{Rasmussen:Williams:06}).
Then by Theorem 14.3 of \cite{Piterbarg:96} and also by Theorem 3.2 of \cite{Konakov:Piterbarg:1984}, there exists $\kappa > 0$ such that 
uniformly in $t$, on any finite interval:
\begin{align*}
&\mathbb{P}\left( a_n \left[ \max_{\mathbf{x} \in \mathcal{I}} \left| \tilde{B}_{n,2}(h_n^{-1} \mathbf{x}) \right| - a_n  \right]  < t \right) = \\
&\exp \left( - 2e^{-t - t^2/2a_n^2} \right)
\sum_{m=0}^{[(d-1)/2]} h_{m,d-1} a_n^{-2m}
\left( 1 + \frac{t}{a_n^2} \right)^{d-2m-1}
+ O( n^{-\kappa})
\end{align*}
as $n \rightarrow \infty$, where 
$a_n$ is obtained as the largest solution to the equation:
\begin{align*}
\frac{ \text{mes}(\mathcal{I}) {h_n}^{-d} \sqrt{\mathrm{det} \Lambda_2} }{(2\pi)^{(d+1)/2} }
a_n^{d-1} e^{-a_n^2/2} = 1,
\end{align*}
$\Lambda_2$ is the covariance matrix of the vector of the first derivative
of the Gaussian field $\tilde{B}_{n,2}$:
\begin{align*}
\Lambda_2  \equiv \mathrm{cov} \;\mathrm{grad} \; \tilde{B}_{n,2}(\mathbf{t}) = 
\left( - \frac{\partial^2 r(0) }{\partial t_i  \partial t_j}, i,j=1,\ldots,d\right),
\end{align*}
and $[\cdot]$ is the integer part of a number. Simple calculation yields 
$\sqrt{\mathrm{det} \Lambda_2} = \lambda^{d/2}$ with $\lambda$ defined in \eqref{def-lambda}.
\end{proof}

\bibliographystyle{econometrica}
\bibliography{reference-dcate}

\begin{table}[t]
\centering
\caption{Monte Carlo results: CATEF estimates, $p=10$}
\label{tab-est-k10}
\scriptsize{
\begin{tabular}{c|cccc|ccc|ccc}
\hline 
& \multicolumn{4}{c|}{DR} & \multicolumn{3}{c|}{IPW} & \multicolumn{3}{c}{RA} \\
\hline
At $\mathbf{x}=$ & BIAS     & SD     & ASE   & RMSE & BIAS     & SD   &  RMSE & BIAS     & SD     &  RMSE \\
\hline \hline
\multicolumn{11}{c}{ $N=500$} \\						
\hline 
\multicolumn{11}{c}{True propensity score model, True regression model} \\
\hline
-1   & -0.002 & 0.213 & 0.195 & 0.213 & -0.001  & 1.234 & 1.234   & -0.000 & 0.140 & 0.140  \\
-0.5 & 0.002  & 0.197 & 0.166 & 0.197 & -0.009  & 1.139 & 1.138   & 0.001  & 0.116 & 0.116  \\
0    & -0.000 & 0.178 & 0.156 & 0.178 & 0.029   & 1.090 & 1.090   & -0.001 & 0.112 & 0.112  \\
0.5  & 0.001  & 0.185 & 0.165 & 0.185 & -0.005  & 1.103 & 1.103   & 0.000  & 0.118 & 0.118  \\
1    & -0.001 & 0.221 & 0.194 & 0.221 & -0.001  & 1.438 & 1.438   & 0.001  & 0.140 & 0.140  \\
\hline 
\multicolumn{11}{c}{True propensity score model, False regression model} \\
\hline
-1   & 0.008  & 0.288 & 0.247 & 0.288 & -0.024  & 1.223 & 1.223   & 0.293  & 0.124 & 0.318  \\
-0.5 & 0.004  & 0.234 & 0.213 & 0.234 & -0.003  & 1.057 & 1.057   & 0.293  & 0.101 & 0.310  \\
0    & 0.003  & 0.216 & 0.202 & 0.216 & 0.015   & 1.028 & 1.028   & 0.293  & 0.092 & 0.307  \\
0.5  & 0.007  & 0.237 & 0.212 & 0.237 & -0.002  & 1.084 & 1.084   & 0.292  & 0.103 & 0.309  \\
1    & 0.008  & 0.286 & 0.247 & 0.286 & -0.013  & 1.321 & 1.321   & 0.292  & 0.127 & 0.318  \\
\hline 
\multicolumn{11}{c}{False propensity score model, True regression model} \\
\hline
-1   & 0.002  & 0.193 & 0.175 & 0.193 & 0.297   & 0.964 & 1.008   & 0.002  & 0.142 & 0.142  \\
-0.5 & 0.002  & 0.161 & 0.147 & 0.161 & 0.293   & 0.833 & 0.883   & 0.000  & 0.118 & 0.118  \\
0    & -0.000 & 0.154 & 0.139 & 0.154 & 0.302   & 0.796 & 0.851   & -0.001 & 0.109 & 0.109  \\
0.5  & -0.003 & 0.159 & 0.146 & 0.159 & 0.271   & 0.853 & 0.895   & -0.002 & 0.116 & 0.116  \\
1    & -0.001 & 0.193 & 0.175 & 0.193 & 0.293   & 1.047 & 1.087   & -0.001 & 0.141 & 0.141  \\
\hline 
\multicolumn{11}{c}{False propensity score model, False regression model} \\
\hline
-1   & 0.298  & 0.201 & 0.185 & 0.360 & 0.312   & 1.003 & 1.051   & 0.294  & 0.125 & 0.320  \\
-0.5 & 0.295  & 0.170 & 0.155 & 0.340 & 0.304   & 0.827 & 0.881   & 0.293  & 0.102 & 0.311  \\
0    & 0.295  & 0.158 & 0.146 & 0.334 & 0.269   & 0.827 & 0.870   & 0.293  & 0.093 & 0.307  \\
0.5  & 0.292  & 0.167 & 0.156 & 0.336 & 0.280   & 0.874 & 0.917   & 0.295  & 0.103 & 0.313  \\
1    & 0.294  & 0.202 & 0.185 & 0.357 & 0.330   & 1.041 & 1.092   & 0.295  & 0.128 & 0.321  \\
\hline \hline
\multicolumn{11}{c}{ $N=2000$} \\
\hline 
\multicolumn{11}{c}{True propensity score model, True regression model} \\
\hline
-1   & -0.001 & 0.113 & 0.102 & 0.113 & 0.009   & 0.610 & 0.610   & -0.000 & 0.072 & 0.072  \\
-0.5 & -0.001 & 0.093 & 0.086 & 0.093 & 0.008   & 0.516 & 0.516   & -0.000 & 0.059 & 0.059  \\
0    & 0.000  & 0.091 & 0.081 & 0.091 & 0.005   & 0.519 & 0.519   & 0.000  & 0.056 & 0.056  \\
0.5  & 0.001  & 0.093 & 0.085 & 0.093 & 0.006   & 0.545 & 0.545   & 0.001  & 0.059 & 0.059  \\
1    & 0.000  & 0.112 & 0.102 & 0.112 & 0.005   & 0.660 & 0.660   & 0.001  & 0.072 & 0.072  \\
\hline 
\multicolumn{11}{c}{True propensity score model, False regression model} \\
\hline
-1   & 0.005  & 0.144 & 0.131 & 0.144 & 0.000   & 0.614 & 0.614   & 0.295  & 0.063 & 0.301  \\
-0.5 & 0.005  & 0.119 & 0.110 & 0.119 & -0.007  & 0.515 & 0.515   & 0.294  & 0.051 & 0.298  \\
0    & -0.001 & 0.124 & 0.106 & 0.124 & 0.008   & 0.512 & 0.512   & 0.293  & 0.047 & 0.297  \\
0.5  & -0.002 & 0.126 & 0.112 & 0.126 & 0.001   & 0.529 & 0.529   & 0.293  & 0.051 & 0.298  \\
1    & -0.004 & 0.144 & 0.132 & 0.144 & 0.005   & 0.643 & 0.642   & 0.293  & 0.064 & 0.299  \\
\hline 
\multicolumn{11}{c}{False propensity score model, True regression model} \\
\hline
-1   & -0.001 & 0.100 & 0.091 & 0.100 & 0.292   & 0.498 & 0.577   & -0.001 & 0.071 & 0.071  \\
-0.5 & -0.001 & 0.083 & 0.076 & 0.083 & 0.286   & 0.434 & 0.520   & -0.002 & 0.060 & 0.060  \\
0    & -0.001 & 0.080 & 0.072 & 0.080 & 0.291   & 0.427 & 0.517   & -0.002 & 0.056 & 0.056  \\
0.5  & -0.002 & 0.085 & 0.076 & 0.085 & 0.296   & 0.451 & 0.539   & -0.001 & 0.060 & 0.060  \\
1    & 0.000  & 0.100 & 0.091 & 0.100 & 0.291   & 0.546 & 0.619   & -0.000 & 0.073 & 0.073  \\
\hline 
\multicolumn{11}{c}{False propensity score model, False regression model} \\
\hline
-1   & 0.291  & 0.103 & 0.095 & 0.309 & 0.294   & 0.503 & 0.582   & 0.292  & 0.063 & 0.299  \\
-0.5 & 0.292  & 0.085 & 0.080 & 0.304 & 0.298   & 0.443 & 0.534   & 0.292  & 0.051 & 0.297  \\
0    & 0.294  & 0.081 & 0.076 & 0.305 & 0.295   & 0.448 & 0.536   & 0.293  & 0.047 & 0.296  \\
0.5  & 0.292  & 0.087 & 0.080 & 0.305 & 0.291   & 0.454 & 0.539   & 0.292  & 0.052 & 0.297  \\
1    & 0.293  & 0.104 & 0.096 & 0.311 & 0.285   & 0.520 & 0.593   & 0.293  & 0.064 & 0.300  \\
\hline
\end{tabular}
}
\end{table}

\begin{table}[t]
\centering
\caption{Monte Carlo results: CATEF estimates, $p=30$}
\label{tab-est-k30}
\scriptsize{
\begin{tabular}{c|cccc|ccc|ccc}
\hline 
& \multicolumn{4}{c|}{DR} & \multicolumn{3}{c|}{IPW} & \multicolumn{3}{c}{RA} \\
\hline
At $\mathbf{x}=$ & BIAS     & SD     & ASE   & RMSE & BIAS     & SD   &  RMSE & BIAS     & SD     &  RMSE \\
\hline \hline
\multicolumn{11}{c}{ $N=500$} \\						
\hline 
\multicolumn{11}{c}{True propensity score model, True regression model} \\
\hline
-1   & -0.001 & 0.221 & 0.211 & 0.221 & 0.005   & 1.327 & 1.326   & -0.002 & 0.149 & 0.149  \\
-0.5 & -0.006 & 0.190 & 0.179 & 0.190 & -0.001  & 1.229 & 1.229   & -0.002 & 0.122 & 0.122  \\
0    & 0.001  & 0.209 & 0.172 & 0.209 & 0.026   & 1.207 & 1.207   & -0.001 & 0.115 & 0.115  \\
0.5  & -0.002 & 0.195 & 0.181 & 0.195 & -0.013  & 1.225 & 1.225   & -0.002 & 0.124 & 0.124  \\
1    & -0.001 & 0.228 & 0.213 & 0.228 & 0.022   & 1.381 & 1.381   & 0.002  & 0.150 & 0.150  \\
\hline 
\multicolumn{11}{c}{True propensity score model, False regression model} \\
\hline
-1   & 0.007  & 0.315 & 0.274 & 0.315 & 0.011   & 1.349 & 1.349   & 0.290  & 0.133 & 0.319  \\
-0.5 & 0.009  & 0.262 & 0.235 & 0.262 & 0.003   & 1.231 & 1.231   & 0.292  & 0.108 & 0.311  \\
0    & 0.009  & 0.247 & 0.222 & 0.247 & 0.008   & 1.161 & 1.161   & 0.293  & 0.097 & 0.308  \\
0.5  & 0.006  & 0.262 & 0.234 & 0.263 & -0.001  & 1.262 & 1.262   & 0.292  & 0.107 & 0.312  \\
1    & 0.014  & 0.303 & 0.270 & 0.304 & -0.022  & 1.414 & 1.414   & 0.292  & 0.132 & 0.320  \\
\hline 
\multicolumn{11}{c}{False propensity score model, True regression model} \\
\hline
-1   & -0.000 & 0.196 & 0.187 & 0.195 & 0.272   & 0.976 & 1.013   & 0.002  & 0.147 & 0.147  \\
-0.5 & 0.000  & 0.161 & 0.157 & 0.161 & 0.284   & 0.856 & 0.902   & 0.001  & 0.123 & 0.123  \\
0    & 0.001  & 0.154 & 0.148 & 0.154 & 0.290   & 0.798 & 0.849   & 0.000  & 0.115 & 0.115  \\
0.5  & -0.001 & 0.164 & 0.156 & 0.164 & 0.305   & 0.845 & 0.899   & -0.001 & 0.125 & 0.125  \\
1    & 0.000  & 0.193 & 0.186 & 0.193 & 0.299   & 1.034 & 1.076   & -0.001 & 0.145 & 0.145  \\
\hline 
\multicolumn{11}{c}{False propensity score model, False regression model} \\
\hline
-1   & 0.291  & 0.205 & 0.197 & 0.356 & 0.289   & 0.964 & 1.006   & 0.291  & 0.135 & 0.321  \\
-0.5 & 0.292  & 0.169 & 0.167 & 0.337 & 0.307   & 0.825 & 0.880   & 0.292  & 0.111 & 0.312  \\
0    & 0.291  & 0.155 & 0.157 & 0.330 & 0.296   & 0.823 & 0.875   & 0.290  & 0.100 & 0.307  \\
0.5  & 0.292  & 0.170 & 0.166 & 0.338 & 0.300   & 0.819 & 0.872   & 0.290  & 0.111 & 0.310  \\
1    & 0.288  & 0.203 & 0.197 & 0.352 & 0.263   & 1.007 & 1.040   & 0.291  & 0.137 & 0.322  \\
\hline \hline
\multicolumn{11}{c}{ $N=2000$} \\
\hline 
\multicolumn{11}{c}{True propensity score model, True regression model} \\
\hline
-1   & 0.003  & 0.106 & 0.102 & 0.106 & -0.002  & 0.613 & 0.613   & 0.000  & 0.074 & 0.074  \\
-0.5 & 0.000  & 0.094 & 0.086 & 0.094 & 0.001   & 0.537 & 0.537   & 0.000  & 0.060 & 0.060  \\
0    & -0.002 & 0.091 & 0.081 & 0.091 & 0.011   & 0.520 & 0.520   & -0.000 & 0.056 & 0.056  \\
0.5  & 0.000  & 0.092 & 0.086 & 0.092 & -0.004  & 0.534 & 0.534   & -0.000 & 0.060 & 0.060  \\
1    & 0.001  & 0.110 & 0.102 & 0.110 & 0.010   & 0.644 & 0.644   & -0.000 & 0.072 & 0.072  \\
\hline 
\multicolumn{11}{c}{True propensity score model, False regression model} \\
\hline
-1   & 0.003  & 0.143 & 0.133 & 0.143 & -0.011  & 0.611 & 0.611   & 0.294  & 0.066 & 0.301  \\
-0.5 & 0.002  & 0.119 & 0.113 & 0.119 & -0.002  & 0.526 & 0.526   & 0.293  & 0.053 & 0.298  \\
0    & 0.001  & 0.122 & 0.108 & 0.122 & 0.019   & 0.523 & 0.523   & 0.293  & 0.049 & 0.297  \\
0.5  & 0.005  & 0.122 & 0.113 & 0.122 & -0.003  & 0.544 & 0.544   & 0.292  & 0.054 & 0.297  \\
1    & 0.004  & 0.144 & 0.132 & 0.144 & 0.010   & 0.622 & 0.622   & 0.291  & 0.066 & 0.299  \\
\hline 
\multicolumn{11}{c}{False propensity score model, True regression model} \\
\hline
-1   & -0.000 & 0.098 & 0.092 & 0.098 & 0.287   & 0.494 & 0.571   & -0.001 & 0.075 & 0.075  \\
-0.5 & 0.001  & 0.081 & 0.077 & 0.081 & 0.296   & 0.429 & 0.521   & -0.001 & 0.061 & 0.061  \\
0    & -0.000 & 0.079 & 0.073 & 0.079 & 0.298   & 0.413 & 0.509   & -0.000 & 0.056 & 0.056  \\
0.5  & -0.002 & 0.083 & 0.077 & 0.083 & 0.286   & 0.446 & 0.529   & -0.002 & 0.060 & 0.060  \\
1    & -0.001 & 0.101 & 0.092 & 0.101 & 0.302   & 0.513 & 0.595   & -0.001 & 0.073 & 0.073  \\
\hline 
\multicolumn{11}{c}{False propensity score model, False regression model} \\
\hline
-1   & 0.294  & 0.104 & 0.098 & 0.312 & 0.301   & 0.493 & 0.578   & 0.292  & 0.066 & 0.300  \\
-0.5 & 0.293  & 0.088 & 0.082 & 0.306 & 0.305   & 0.431 & 0.528   & 0.293  & 0.054 & 0.298  \\
0    & 0.292  & 0.081 & 0.077 & 0.303 & 0.292   & 0.418 & 0.510   & 0.293  & 0.048 & 0.297  \\
0.5  & 0.291  & 0.087 & 0.082 & 0.304 & 0.287   & 0.436 & 0.522   & 0.292  & 0.053 & 0.297  \\
1    & 0.290  & 0.106 & 0.098 & 0.309 & 0.292   & 0.518 & 0.595   & 0.291  & 0.067 & 0.299  \\
\hline
\end{tabular}
}
\end{table}

\begin{table}[t]
\caption{Monte Carlo results, CATEF confidence band, $p=10$}
\label{tab-cb-k10}
\begin{tabular}{c|cccc}
\hline
Confidence level & CP & Mcri    & Sdcri & GCP \\
\hline \hline 
& \multicolumn{4}{c}{$N=500$} \\
\hline 
\multicolumn{5}{c}{True propensity score model, True regression model} \\
\hline
99\% & 0.986 & 3.290 & 0.091 & 0.999     \\
95\% & 0.939 & 2.750 & 0.107 & 0.999     \\
90\% & 0.887 & 2.474 & 0.118 & 0.995     \\
\hline 
\multicolumn{5}{c}{True propensity score model, False regression model} \\
\hline
99\% & 0.953 & 3.300 & 0.112 & 0.998     \\
95\% & 0.881 & 2.762 & 0.131 & 0.991     \\
90\% & 0.823 & 2.487 & 0.144 & 0.979     \\
\hline 
\multicolumn{5}{c}{False propensity score model, True regression model} \\
\hline
99\% & 0.980 & 3.284 & 0.075 & 1.000     \\
95\% & 0.926 & 2.743 & 0.088 & 0.998     \\
90\% & 0.856 & 2.466 & 0.097 & 0.994     \\
\hline 
\multicolumn{5}{c}{False propensity score model, False regression model} \\
\hline
99\% & 0.561 & 3.285 & 0.079 & 0.999     \\
95\% & 0.299 & 2.744 & 0.093 & 0.988     \\
90\% & 0.186 & 2.468 & 0.102 & 0.965     \\
\hline \hline 
& \multicolumn{4}{c}{$N=3000$} \\
\hline 
\multicolumn{5}{c}{True propensity score model, True regression model} \\
\hline
99\% & 0.988 & 3.299 & 0.090 & 1.000     \\
95\% & 0.941 & 2.761 & 0.106 & 1.000     \\
90\% & 0.880 & 2.486 & 0.117 & 0.997     \\
\hline 
\multicolumn{5}{c}{True propensity score model, False regression model} \\
\hline
99\% & 0.966 & 3.303 & 0.099 & 1.000     \\
95\% & 0.907 & 2.765 & 0.116 & 0.993     \\
90\% & 0.848 & 2.491 & 0.128 & 0.985     \\
\hline 
\multicolumn{5}{c}{False propensity score model, True regression model} \\
\hline
99\% & 0.987 & 3.296 & 0.082 & 1.000     \\
95\% & 0.929 & 2.757 & 0.097 & 0.998     \\
90\% & 0.863 & 2.482 & 0.106 & 0.995     \\
\hline 
\multicolumn{5}{c}{False propensity score model, False regression model} \\
\hline
99\% & 0.014 & 3.296 & 0.081 & 0.986     \\
95\% & 0.001 & 2.756 & 0.096 & 0.912     \\
90\% & 0.000 & 2.481 & 0.105 & 0.843     \\
\hline
\end{tabular}
\end{table}

\begin{table}[t]
\caption{Monte Carlo results, CATEF confidence band, $p=30$}
\label{tab-cb-k30}
\begin{tabular}{c|cccc}
\hline
Confidence level & CP & Mcri    & Sdcri & GCP \\
\hline \hline 
& \multicolumn{4}{c}{$N=500$} \\
\hline 
\multicolumn{5}{c}{True propensity score model, True regression model} \\
\hline
99\% & 0.993 & 3.294 & 0.096 & 1.000     \\
95\% & 0.960 & 2.754 & 0.113 & 0.999     \\
90\% & 0.915 & 2.478 & 0.124 & 0.997     \\
\hline 
\multicolumn{5}{c}{True propensity score model, False regression model} \\
\hline
99\% & 0.969 & 3.302 & 0.113 & 0.999     \\
95\% & 0.909 & 2.763 & 0.132 & 0.994     \\
90\% & 0.853 & 2.489 & 0.145 & 0.985     \\
\hline 
\multicolumn{5}{c}{False propensity score model, True regression model} \\
\hline
99\% & 0.989 & 3.284 & 0.078 & 1.000     \\
95\% & 0.952 & 2.742 & 0.092 & 0.999     \\
90\% & 0.904 & 2.466 & 0.101 & 0.996     \\
\hline 
\multicolumn{5}{c}{False propensity score model, False regression model} \\
\hline
99\% & 0.683 & 3.284 & 0.077 & 0.999     \\
95\% & 0.410 & 2.743 & 0.091 & 0.994     \\
90\% & 0.264 & 2.466 & 0.099 & 0.977     \\
\hline \hline 
& \multicolumn{4}{c}{$N=2000$} \\
\hline 
\multicolumn{5}{c}{True propensity score model, True regression model} \\
\hline
99\% & 0.990 & 3.296 & 0.085 & 1.000     \\
95\% & 0.950 & 2.757 & 0.100 & 1.000     \\
90\% & 0.893 & 2.481 & 0.110 & 0.998     \\
\hline 
\multicolumn{5}{c}{True propensity score model, False regression model} \\
\hline
99\% & 0.975 & 3.304 & 0.101 & 0.999     \\
95\% & 0.920 & 2.766 & 0.119 & 0.996     \\
90\% & 0.863 & 2.492 & 0.130 & 0.988     \\
\hline 
\multicolumn{5}{c}{False propensity score model, True regression model} \\
\hline
99\% & 0.987 & 3.296 & 0.079 & 1.000     \\
95\% & 0.935 & 2.756 & 0.093 & 0.998     \\
90\% & 0.876 & 2.481 & 0.103 & 0.993     \\
\hline 
\multicolumn{5}{c}{False propensity score model, False regression model} \\
\hline
99\% & 0.021 & 3.296 & 0.081 & 0.990     \\
95\% & 0.003 & 2.757 & 0.096 & 0.921     \\
90\% & 0.001 & 2.482 & 0.106 & 0.852     \\
\hline
\end{tabular}
\end{table}

\setcounter{figure}{0} 

\begin{figure}[p]
\label{fig-bw-cb5p}
\caption{CATEF for the effect of smoking on birth weights, Pennsylvania data, 95\% confidence bands} 
\includegraphics[width=0.8\linewidth, bb = 0 0 648 648]{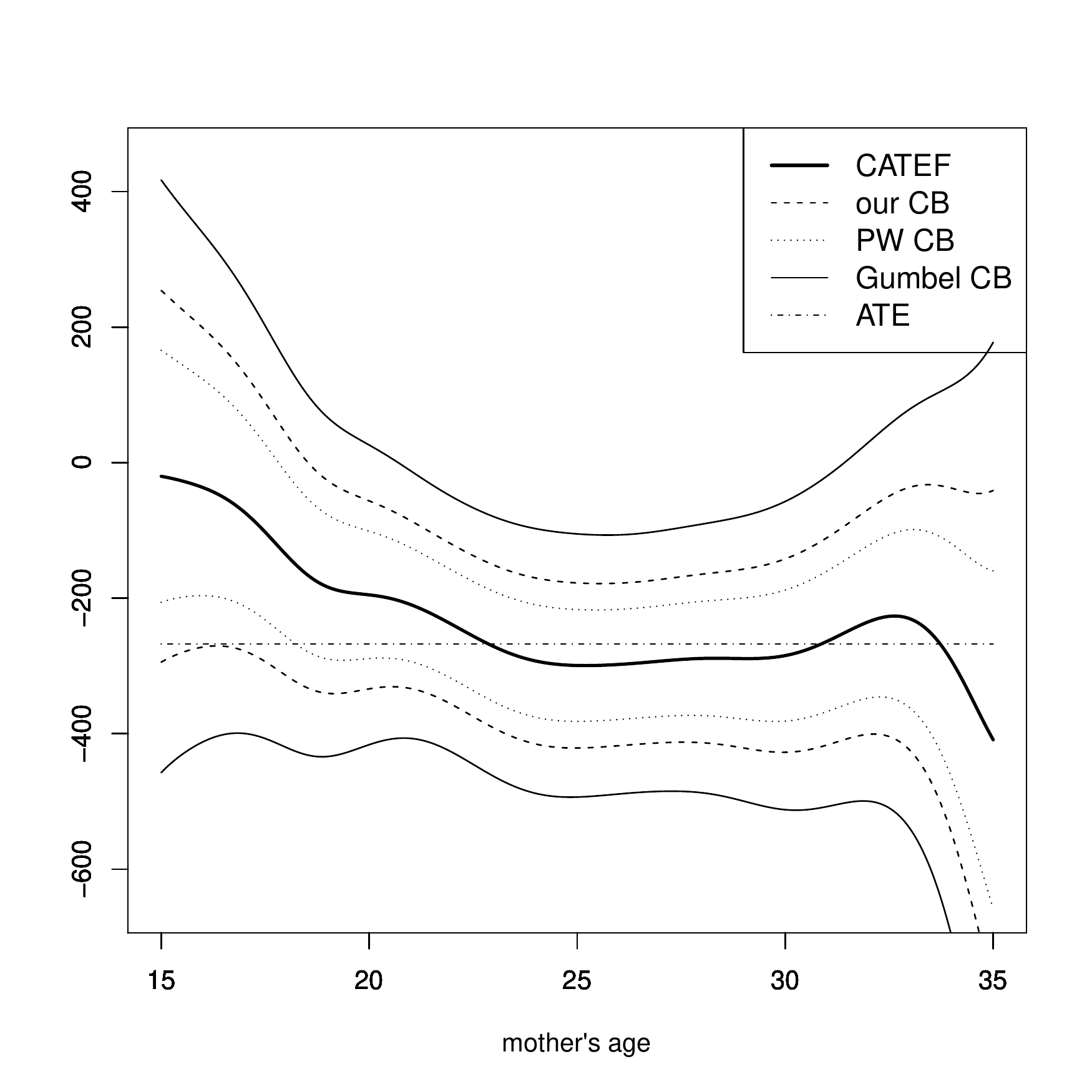}

\begin{minipage}{\linewidth}
Note: ``CATEF'' = the estimated CATEF; 
``our CB'' = the uniformly valid confidence band proposed in this paper; 
 ``PW CB'' = the confidence band that is valid only in a pointwise sense;
``Gumbel CB'' = the uniformly valid confidence band based on the Gumbel approximation;
``ATE'' = the estimated value of the average treatment effect.
\end{minipage}
\end{figure}

\begin{figure}[p]
\label{fig-bw-cb5p-white}
\caption{CATEF for the effect of smoking on birth weights, North Carolina data, with 95\% confidence bands} 
\includegraphics[width=0.8\linewidth, bb = 0 0 504 504]{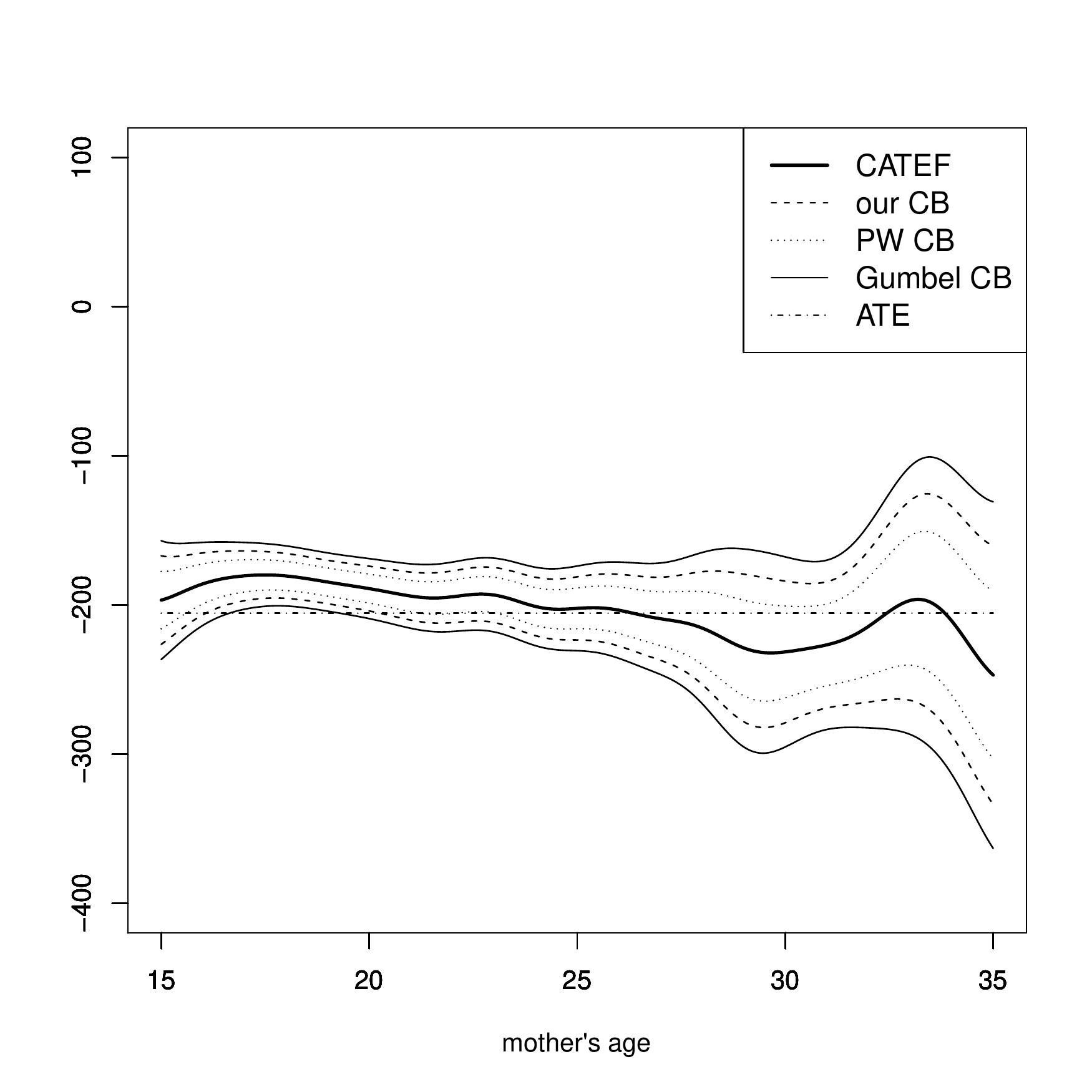}

\begin{minipage}{\linewidth}
Note: ``CATEF'' = the estimated CATEF; 
``our CB'' = the uniformly valid confidence band proposed in this paper; 
 ``PW CB'' = the confidence band that is valid only in a pointwise sense;
``Gumbel CB'' = the uniformly valid confidence band based on the Gumbel approximation;
``ATE'' = the estimated value of the average treatment effect.
\end{minipage}
\end{figure}

\begin{figure}[p]
\label{fig-bw-cb5p-oldcov}
\caption{CATEF for the effect of smoking on birth weights, North Carolina data, smaller set of covariates, with 95\% confidence bands} 
\includegraphics[width=0.8\linewidth, bb = 0 0 504 504]{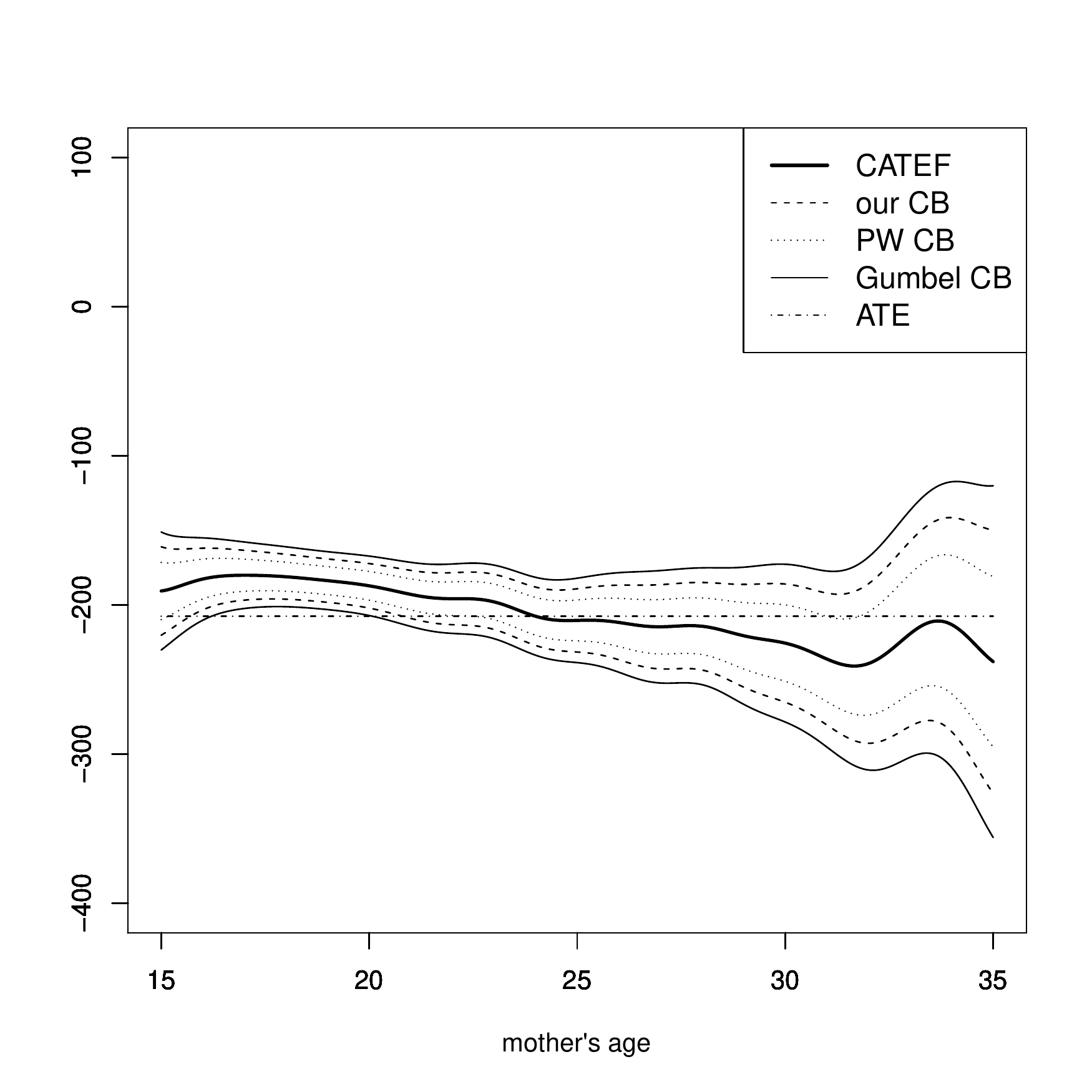}

\begin{minipage}{\linewidth}
Note: ``CATEF'' = the estimated CATEF; 
``our CB'' = the uniformly valid confidence band proposed in this paper; 
 ``PW CB'' = the confidence band that is valid only in a pointwise sense;
``Gumbel CB'' = the uniformly valid confidence band based on the Gumbel approximation;
``ATE'' = the estimated value of the average treatment effect.
\end{minipage}
\end{figure}

\end{document}